\documentclass[onecolumn,11pt]{article}
\usepackage[top=2.7cm, bottom=2.7cm, left=2.7cm, right=2.7cm,a4paper]{geometry}

\usepackage[utf8]{inputenc} 
\usepackage{amsmath,amsthm}
\usepackage{bm}  
\usepackage{bbm}
\usepackage{verbatim}
\usepackage{xcolor}
\usepackage{cases}
\usepackage[titletoc,title]{appendix}
\usepackage{nicefrac} 

\usepackage[backend=bibtex,sorting=none,style=phys,biblabel=brackets,backref=false,bibencoding=utf8,maxnames=20,eprint=true]{biblatex}
\makeatletter
\pretocmd{\blx@head@bibintoc}{\phantomsection}{}{\ddt}
\makeatother
\DefineBibliographyStrings{english}{%
 backrefpage = {page},
 backrefpages = {pages},
}
\makeatletter
\def\blx@bblfile@bibtex{%
 \blx@secinit
 \begingroup
 \blx@bblstart
\begingroup
\makeatletter
\@ifundefined{ver@biblatex.sty}
  {\@latex@error
     {Missing 'biblatex' package}
     {The bibliography requires the 'biblatex' package.}
      \aftergroup }
  {}
\endgroup

\entry{tomamichel_quantum_2016}{book}{}
  \name{author}{1}{}{%
    {{}%
     {Tomamichel}{T.}%
     {Marco}{M.}%
     {}{}%
     {}{}}%
  }
  \list{publisher}{1}{%
    {{Springer International Publishing}}%
  }
  \strng{namehash}{TM1}
  \strng{fullhash}{TM1}
  \verb{eprint}
  \verb 1504.00233
  \endverb
  \field{isbn}{978-3-319-21890-8 978-3-319-21891-5}
  \field{series}{SpringerBriefs in Mathematical Physics}
  \field{title}{Quantum Information Processing with Finite Resources}
  \verb{url}
  \verb http://link.springer.com/10.1007/978-3-319-21891-5
  \endverb
  \field{volume}{5}
  \list{location}{1}{%
    {{Cham}}%
  }
  \field{eprinttype}{arxiv}
  \field{eprintclass}{quant-ph}
  \field{year}{2016}
  \field{urlday}{25}
  \field{urlmonth}{11}
  \field{urlyear}{2015}
  \warn{\item Can't use 'eprinttype' + 'archiveprefix'}
\endentry

\entry{petz_quasi-entropies_1986}{article}{}
  \name{author}{1}{}{%
    {{}%
     {Petz}{P.}%
     {D{\'e}nes}{D.}%
     {}{}%
     {}{}}%
  }
  \strng{namehash}{PD1}
  \strng{fullhash}{PD1}
  \verb{doi}
  \verb 10.1016/0034-4877(86)90067-4
  \endverb
  \field{issn}{0034-4877}
  \field{number}{1}
  \field{pages}{57\bibrangedash 65}
  \field{title}{Quasi-entropies for finite quantum systems}
  \verb{url}
  \verb http://www.sciencedirect.com/science/article/pii/0034487786900674
  \endverb
  \field{volume}{23}
  \field{journaltitle}{Reports on Mathematical Physics}
  \field{month}{02}
  \field{year}{1986}
\endentry

\entry{muller-lennert_quantum_2013}{article}{}
  \name{author}{5}{}{%
    {{}%
     {M{\"u}ller-Lennert}{M.-L.}%
     {Martin}{M.}%
     {}{}%
     {}{}}%
    {{}%
     {Dupuis}{D.}%
     {Fr{\'e}d{\'e}ric}{F.}%
     {}{}%
     {}{}}%
    {{}%
     {Szehr}{S.}%
     {Oleg}{O.}%
     {}{}%
     {}{}}%
    {{}%
     {Fehr}{F.}%
     {Serge}{S.}%
     {}{}%
     {}{}}%
    {{}%
     {Tomamichel}{T.}%
     {Marco}{M.}%
     {}{}%
     {}{}}%
  }
  \strng{namehash}{MLMDFSOFSTM1}
  \strng{fullhash}{MLMDFSOFSTM1}
  \verb{doi}
  \verb 10.1063/1.4838856
  \endverb
  \verb{eprint}
  \verb 1306.3142
  \endverb
  \field{issn}{0022-2488, 1089-7658}
  \field{number}{12}
  \field{pages}{21}
  \field{shorttitle}{On quantum {{R{\'e}nyi}} entropies}
  \field{title}{On quantum {{R{\'e}nyi}} entropies: A new generalization and
  some properties}
  \verb{url}
  \verb http://scitation.aip.org/content/aip/journal/jmp/54/12/10.1063/1.483885
  \verb 6
  \endverb
  \field{volume}{54}
  \field{journaltitle}{Journal of Mathematical Physics}
  \field{eprinttype}{arxiv}
  \field{eprintclass}{quant-ph}
  \field{day}{01}
  \field{month}{12}
  \field{year}{2013}
  \field{urlday}{12}
  \field{urlmonth}{03}
  \field{urlyear}{2015}
  \warn{\item Can't use 'eprinttype' + 'archiveprefix'}
\endentry

\entry{wilde_strong_2014}{article}{}
  \name{author}{3}{}{%
    {{}%
     {Wilde}{W.}%
     {Mark~M.}{M.~M.}%
     {}{}%
     {}{}}%
    {{}%
     {Winter}{W.}%
     {Andreas}{A.}%
     {}{}%
     {}{}}%
    {{}%
     {Yang}{Y.}%
     {Dong}{D.}%
     {}{}%
     {}{}}%
  }
  \strng{namehash}{WMMWAYD1}
  \strng{fullhash}{WMMWAYD1}
  \verb{doi}
  \verb 10.1007/s00220-014-2122-x
  \endverb
  \verb{eprint}
  \verb 1306.1586
  \endverb
  \field{issn}{0010-3616, 1432-0916}
  \field{number}{2}
  \field{pages}{593\bibrangedash 622}
  \field{shortjournal}{Commun. Math. Phys.}
  \field{title}{Strong Converse for the Classical Capacity of
  Entanglement-Breaking and {{Hadamard}} Channels via a Sandwiched
  {{R{\'e}nyi}} Relative Entropy}
  \verb{url}
  \verb http://link.springer.com/article/10.1007/s00220-014-2122-x
  \endverb
  \field{volume}{331}
  \field{langid}{english}
  \field{journaltitle}{Communications in Mathematical Physics}
  \field{eprinttype}{arxiv}
  \field{eprintclass}{quant-ph}
  \field{day}{15}
  \field{month}{07}
  \field{year}{2014}
  \field{urlday}{12}
  \field{urlmonth}{03}
  \field{urlyear}{2015}
  \warn{\item Can't use 'eprinttype' + 'archiveprefix'}
\endentry

\entry{lieb_inequalities_1976}{incollection}{}
  \name{author}{2}{}{%
    {{}%
     {Lieb}{L.}%
     {E.~H.}{E.~H.}%
     {}{}%
     {}{}}%
    {{}%
     {Thirring}{T.}%
     {Walter}{W.}%
     {}{}%
     {}{}}%
  }
  \list{publisher}{1}{%
    {{Princeton University Press}}%
  }
  \strng{namehash}{LEHTW1}
  \strng{fullhash}{LEHTW1}
  \field{booktitle}{Studies in {{Mathematical Physics}}: {{Essays}} in
  {{Honor}} of {{Valentine Bargmann}}}
  \field{pages}{269\bibrangedash 304}
  \field{series}{Princeton Series in Physics}
  \field{title}{Inequalities for the Moments of the Eigenvalues of the
  {{Schr{\"o}dinger}} {{Hamiltonian}} and Their Relation to {{Sobolev}}
  Inequalities}
  \verb{url}
  \verb http://www.jstor.org/stable/j.ctt13x134j.16
  \endverb
  \field{year}{1976}
\endentry

\entry{araki_inequality_1990}{article}{}
  \name{author}{1}{}{%
    {{}%
     {Araki}{A.}%
     {Huzihiro}{H.}%
     {}{}%
     {}{}}%
  }
  \strng{namehash}{AH1}
  \strng{fullhash}{AH1}
  \verb{doi}
  \verb 10.1007/BF01045887
  \endverb
  \field{issn}{0377-9017, 1573-0530}
  \field{number}{2}
  \field{pages}{167\bibrangedash 170}
  \field{shortjournal}{Lett Math Phys}
  \field{title}{On an inequality of {{Lieb}} and {{Thirring}}}
  \verb{url}
  \verb http://link.springer.com/article/10.1007/BF01045887
  \endverb
  \field{volume}{19}
  \field{langid}{english}
  \field{journaltitle}{Letters in Mathematical Physics}
  \field{month}{02}
  \field{year}{1990}
  \field{urlday}{05}
  \field{urlmonth}{07}
  \field{urlyear}{2016}
\endentry

\entry{barnum_reversing_2002}{article}{}
  \name{author}{2}{}{%
    {{}%
     {Barnum}{B.}%
     {H.}{H.}%
     {}{}%
     {}{}}%
    {{}%
     {Knill}{K.}%
     {E.}{E.}%
     {}{}%
     {}{}}%
  }
  \strng{namehash}{BHKE1}
  \strng{fullhash}{BHKE1}
  \verb{eprint}
  \verb quant-ph/0004088
  \endverb
  \field{number}{5}
  \field{pages}{2097\bibrangedash 2106}
  \field{shortjournal}{J. Math. Phys.}
  \field{title}{Reversing quantum dynamics with near-optimal quantum and
  classical fidelity}
  \verb{url}
  \verb http://link.aip.org/link/?JMP/43/2097/1
  \endverb
  \field{volume}{43}
  \field{journaltitle}{Journal of Mathematical Physics}
  \field{eprinttype}{arxiv}
  \field{month}{05}
  \field{year}{2002}
  \field{urlday}{14}
  \field{urlmonth}{02}
  \field{urlyear}{2008}
  \warn{\item Can't use 'eprinttype' + 'archiveprefix'}
\endentry

\entry{dupuis_entanglement_2015}{article}{}
  \name{author}{3}{}{%
    {{}%
     {Dupuis}{D.}%
     {F.}{F.}%
     {}{}%
     {}{}}%
    {{}%
     {Fawzi}{F.}%
     {O.}{O.}%
     {}{}%
     {}{}}%
    {{}%
     {Wehner}{W.}%
     {S.}{S.}%
     {}{}%
     {}{}}%
  }
  \strng{namehash}{DFFOWS1}
  \strng{fullhash}{DFFOWS1}
  \verb{doi}
  \verb 10.1109/TIT.2014.2371464
  \endverb
  \verb{eprint}
  \verb 1305.1316
  \endverb
  \field{issn}{0018-9448}
  \field{number}{2}
  \field{pages}{1093\bibrangedash 1112}
  \field{title}{Entanglement Sampling and Applications}
  \field{volume}{61}
  \field{journaltitle}{IEEE Transactions on Information Theory}
  \field{eprinttype}{arxiv}
  \field{eprintclass}{quant-ph}
  \field{month}{02}
  \field{year}{2015}
  \warn{\item Can't use 'eprinttype' + 'archiveprefix'}
\endentry

\entry{bhatia_matrix_1997}{book}{}
  \name{author}{1}{}{%
    {{}%
     {Bhatia}{B.}%
     {Rajendra}{R.}%
     {}{}%
     {}{}}%
  }
  \list{publisher}{1}{%
    {{Springer}}%
  }
  \strng{namehash}{BR1}
  \strng{fullhash}{BR1}
  \field{isbn}{978-1-4612-6857-4 978-1-4612-0653-8}
  \field{series}{Graduate Texts in Mathematics}
  \field{title}{Matrix {{Analysis}}}
  \verb{url}
  \verb http://link.springer.com/10.1007/978-1-4612-0653-8
  \endverb
  \field{volume}{169}
  \list{location}{1}{%
    {{New York}}%
  }
  \field{year}{1997}
  \field{urlday}{28}
  \field{urlmonth}{07}
  \field{urlyear}{2016}
\endentry

\entry{audenaert_araki-lieb-thirring_2008}{article}{}
  \name{author}{1}{}{%
    {{}%
     {Audenaert}{A.}%
     {Koenraad M.~R.}{K.~M.~R.}%
     {}{}%
     {}{}}%
  }
  \strng{namehash}{AKMR1}
  \strng{fullhash}{AKMR1}
  \verb{eprint}
  \verb math/0701129
  \endverb
  \field{number}{1}
  \field{pages}{78\bibrangedash 83}
  \field{title}{On the {{Araki-Lieb-Thirring}} inequality}
  \verb{url}
  \verb http://www.math.ualberta.ca/ijiss/SS-Volume-4-2008/No-1-08/SS-08-01-08.
  \verb pdf
  \endverb
  \field{volume}{4}
  \field{journaltitle}{International Journal of Information and Systems
  Sciences}
  \field{eprinttype}{arxiv}
  \field{year}{2008}
  \field{urlday}{15}
  \field{urlmonth}{05}
  \field{urlyear}{2016}
  \warn{\item Can't use 'eprinttype' + 'archiveprefix'}
\endentry

\entry{ando94}{article}{}
  \name{author}{1}{}{%
    {{}%
     {Ando}{A.}%
     {T.}{T.}%
     {}{}%
     {}{}}%
  }
  \strng{namehash}{AT1}
  \strng{fullhash}{AT1}
  \verb{doi}
  \verb 10.1016/0024-3795(94)90341-7
  \endverb
  \field{issn}{0024-3795}
  \field{pages}{17 \bibrangedash  67}
  \field{title}{Majorizations and inequalities in matrix theory}
  \verb{url}
  \verb http://www.sciencedirect.com/science/article/pii/0024379594903417
  \endverb
  \field{volume}{199}
  \field{journaltitle}{Linear Algebra and its Applications}
  \field{year}{1994}
\endentry

\entry{hiai_equality_1994}{article}{}
  \name{author}{1}{}{%
    {{}%
     {Hiai}{H.}%
     {Fumio}{F.}%
     {}{}%
     {}{}}%
  }
  \strng{namehash}{HF1}
  \strng{fullhash}{HF1}
  \verb{doi}
  \verb 10.1080/03081089408818297
  \endverb
  \field{issn}{0308-1087}
  \field{number}{4}
  \field{pages}{239\bibrangedash 249}
  \field{title}{Equality cases in matrix norm inequalities of
  {{Golden-Thompson}} type}
  \verb{url}
  \verb http://dx.doi.org/10.1080/03081089408818297
  \endverb
  \field{volume}{36}
  \field{journaltitle}{Linear and Multilinear Algebra}
  \field{day}{01}
  \field{month}{04}
  \field{year}{1994}
  \field{urlday}{29}
  \field{urlmonth}{07}
  \field{urlyear}{2016}
\endentry

\entry{mosonyi_coding_2015}{article}{}
  \name{author}{1}{}{%
    {{}%
     {Mosonyi}{M.}%
     {Milan}{M.}%
     {}{}%
     {}{}}%
  }
  \strng{namehash}{MM1}
  \strng{fullhash}{MM1}
  \verb{doi}
  \verb 10.1109/TIT.2015.2417877
  \endverb
  \verb{eprint}
  \verb 1310.7525
  \endverb
  \field{issn}{0018-9448, 1557-9654}
  \field{number}{6}
  \field{pages}{2997\bibrangedash 3012}
  \field{title}{Coding Theorems for Compound Problems via Quantum {{R{\'e}nyi}}
  Divergences}
  \verb{url}
  \verb http://ieeexplore.ieee.org/lpdocs/epic03/wrapper.htm?arnumber=7086060
  \endverb
  \field{volume}{61}
  \field{journaltitle}{IEEE Transactions on Information Theory}
  \field{eprinttype}{arxiv}
  \field{eprintclass}{quant-ph}
  \field{month}{06}
  \field{year}{2015}
  \field{urlday}{05}
  \field{urlmonth}{07}
  \field{urlyear}{2016}
  \warn{\item Can't use 'eprinttype' + 'archiveprefix'}
\endentry

\entry{renner_smooth_2004}{inproceedings}{}
  \name{author}{2}{}{%
    {{}%
     {Renner}{R.}%
     {R.}{R.}%
     {}{}%
     {}{}}%
    {{}%
     {Wolf}{W.}%
     {S.}{S.}%
     {}{}%
     {}{}}%
  }
  \strng{namehash}{RRWS1}
  \strng{fullhash}{RRWS1}
  \field{booktitle}{Proceedings of the 2004 {{International Symposium on
  Information Theory}} ({{ISIT}})}
  \verb{doi}
  \verb 10.1109/ISIT.2004.1365269
  \endverb
  \field{isbn}{0-7803-8280-3}
  \field{pages}{233}
  \field{title}{Smooth {{R\'enyi}} entropy and applications}
  \verb{url}
  \verb http://ieeexplore.ieee.org/xpl/freeabs_all.jsp?arnumber=1365269
  \endverb
  \field{year}{2004}
\endentry

\entry{tomamichel_fully_2009}{article}{}
  \name{author}{3}{}{%
    {{}%
     {Tomamichel}{T.}%
     {Marco}{M.}%
     {}{}%
     {}{}}%
    {{}%
     {Colbeck}{C.}%
     {Roger}{R.}%
     {}{}%
     {}{}}%
    {{}%
     {Renner}{R.}%
     {Renato}{R.}%
     {}{}%
     {}{}}%
  }
  \strng{namehash}{TMCRRR1}
  \strng{fullhash}{TMCRRR1}
  \verb{doi}
  \verb 10.1109/TIT.2009.2032797
  \endverb
  \verb{eprint}
  \verb 0811.1221
  \endverb
  \field{number}{12}
  \field{pages}{5840\bibrangedash 5847}
  \field{title}{A Fully Quantum Asymptotic Equipartition Property}
  \field{volume}{55}
  \field{journaltitle}{IEEE Transactions on Information Theory}
  \field{eprinttype}{arxiv}
  \field{eprintclass}{quant-ph}
  \field{day}{17}
  \field{month}{11}
  \field{year}{2009}
  \warn{\item Can't use 'eprinttype' + 'archiveprefix'}
\endentry

\entry{tomamichel_relating_2014}{article}{}
  \name{author}{3}{}{%
    {{}%
     {Tomamichel}{T.}%
     {Marco}{M.}%
     {}{}%
     {}{}}%
    {{}%
     {Berta}{B.}%
     {Mario}{M.}%
     {}{}%
     {}{}}%
    {{}%
     {Hayashi}{H.}%
     {Masahito}{M.}%
     {}{}%
     {}{}}%
  }
  \strng{namehash}{TMBMHM1}
  \strng{fullhash}{TMBMHM1}
  \verb{doi}
  \verb 10.1063/1.4892761
  \endverb
  \verb{eprint}
  \verb 1311.3887
  \endverb
  \field{issn}{0022-2488, 1089-7658}
  \field{number}{8}
  \field{pages}{082206}
  \field{title}{Relating different quantum generalizations of the conditional
  {{R{\'e}nyi}} entropy}
  \verb{url}
  \verb http://scitation.aip.org/content/aip/journal/jmp/55/8/10.1063/1.4892761
  \endverb
  \field{volume}{55}
  \field{journaltitle}{Journal of Mathematical Physics}
  \field{eprinttype}{arxiv}
  \field{eprintclass}{quant-ph}
  \field{day}{01}
  \field{month}{08}
  \field{year}{2014}
  \field{urlday}{22}
  \field{urlmonth}{09}
  \field{urlyear}{2014}
  \warn{\item Can't use 'eprinttype' + 'archiveprefix'}
\endentry

\entry{beigi_sandwiched_2013}{article}{}
  \name{author}{1}{}{%
    {{}%
     {Beigi}{B.}%
     {Salman}{S.}%
     {}{}%
     {}{}}%
  }
  \strng{namehash}{BS1}
  \strng{fullhash}{BS1}
  \verb{doi}
  \verb 10.1063/1.4838855
  \endverb
  \verb{eprint}
  \verb 1306.5920
  \endverb
  \field{issn}{0022-2488, 1089-7658}
  \field{number}{12}
  \field{pages}{122202}
  \field{title}{Sandwiched {{R{\'e}nyi}} divergence satisfies data processing
  inequality}
  \verb{url}
  \verb http://scitation.aip.org/content/aip/journal/jmp/54/12/10.1063/1.483885
  \verb 5
  \endverb
  \field{volume}{54}
  \field{journaltitle}{Journal of Mathematical Physics}
  \field{eprinttype}{arxiv}
  \field{eprintclass}{quant-ph}
  \field{day}{01}
  \field{month}{12}
  \field{year}{2013}
  \field{urlday}{12}
  \field{urlmonth}{03}
  \field{urlyear}{2015}
  \warn{\item Can't use 'eprinttype' + 'archiveprefix'}
\endentry

\entry{konig_operational_2009}{article}{}
  \name{author}{3}{}{%
    {{}%
     {K{\"o}nig}{K.}%
     {Robert}{R.}%
     {}{}%
     {}{}}%
    {{}%
     {Renner}{R.}%
     {Renato}{R.}%
     {}{}%
     {}{}}%
    {{}%
     {Schaffner}{S.}%
     {Christian}{C.}%
     {}{}%
     {}{}}%
  }
  \strng{namehash}{KRRRSC1}
  \strng{fullhash}{KRRRSC1}
  \verb{doi}
  \verb 10.1109/TIT.2009.2025545
  \endverb
  \verb{eprint}
  \verb 0807.1338
  \endverb
  \field{issn}{0018-9448}
  \field{number}{9}
  \field{pages}{4337\bibrangedash 4347}
  \field{shortjournal}{Information Theory, IEEE Transactions on}
  \field{title}{The Operational Meaning of Min- and Max-Entropy}
  \field{volume}{55}
  \field{journaltitle}{Information Theory, IEEE Transactions on}
  \field{eprinttype}{arxiv}
  \field{eprintclass}{quant-ph}
  \field{year}{2009}
  \warn{\item Can't use 'eprinttype' + 'archiveprefix'}
\endentry

\entry{berta_single-shot_2008}{thesis}{}
  \name{author}{1}{}{%
    {{}%
     {Berta}{B.}%
     {Mario}{M.}%
     {}{}%
     {}{}}%
  }
  \strng{namehash}{BM1}
  \strng{fullhash}{BM1}
  \verb{eprint}
  \verb 0912.4495
  \endverb
  \field{title}{Single-Shot Quantum State Merging}
  \verb{url}
  \verb http://arxiv.org/abs/0912.4495v1
  \endverb
  \list{institution}{1}{%
    {{ETH Z{\"u}rich}}%
  }
  \field{type}{Diplom {{Thesis}}}
  \field{eprinttype}{arxiv}
  \field{eprintclass}{quant-ph}
  \field{year}{2008}
\endentry

\entry{luo_informational_2004}{article}{}
  \name{author}{2}{}{%
    {{}%
     {Luo}{L.}%
     {Shunlong}{S.}%
     {}{}%
     {}{}}%
    {{}%
     {Zhang}{Z.}%
     {Qiang}{Q.}%
     {}{}%
     {}{}}%
  }
  \strng{namehash}{LSZQ1}
  \strng{fullhash}{LSZQ1}
  \verb{doi}
  \verb 10.1103/PhysRevA.69.032106
  \endverb
  \field{number}{3}
  \field{pages}{032106}
  \field{shortjournal}{Phys. Rev. A}
  \field{title}{Informational Distance on Quantum-State Space}
  \verb{url}
  \verb http://link.aps.org/doi/10.1103/PhysRevA.69.032106
  \endverb
  \field{volume}{69}
  \field{journaltitle}{Physical Review A}
  \field{day}{16}
  \field{month}{03}
  \field{year}{2004}
  \field{urlday}{22}
  \field{urlmonth}{03}
  \field{urlyear}{2016}
\endentry

\entry{winter_extrinsic_2004}{article}{}
  \name{author}{1}{}{%
    {{}%
     {Winter}{W.}%
     {Andreas}{A.}%
     {}{}%
     {}{}}%
  }
  \strng{namehash}{WA1}
  \strng{fullhash}{WA1}
  \verb{doi}
  \verb 10.1007/s00220-003-0989-z
  \endverb
  \verb{eprint}
  \verb quant-ph/0109050
  \endverb
  \field{issn}{0010-3616, 1432-0916}
  \field{number}{1}
  \field{pages}{157\bibrangedash 185}
  \field{title}{"{{Extrinsic}}" and "Intrinsic" Data in Quantum Measurements:
  Asymptotic Convex Decomposition of Positive Operator Valued Measures}
  \verb{url}
  \verb http://www.springerlink.com/content/4yeebnxm2y1qnckl/
  \endverb
  \field{volume}{244}
  \field{journaltitle}{Communications in Mathematical Physics}
  \field{eprinttype}{arxiv}
  \field{day}{01}
  \field{month}{01}
  \field{year}{2004}
  \field{urlday}{13}
  \field{urlmonth}{03}
  \field{urlyear}{2012}
  \warn{\item Can't use 'eprinttype' + 'archiveprefix'}
\endentry

\entry{uhlmann_transition_1976}{article}{}
  \name{author}{1}{}{%
    {{}%
     {Uhlmann}{U.}%
     {A.}{A.}%
     {}{}%
     {}{}}%
  }
  \strng{namehash}{UA1}
  \strng{fullhash}{UA1}
  \verb{doi}
  \verb 10.1016/0034-4877(76)90060-4
  \endverb
  \field{number}{2}
  \field{pages}{273\bibrangedash 279}
  \field{title}{The "transition probability" in the state space of a *-algebra}
  \verb{url}
  \verb http://www.sciencedirect.com/science/article/B6VN0-45FCR8M-1H/2/a63a0ee
  \verb 13819a2a70bf4b4eb8b704235
  \endverb
  \field{volume}{9}
  \field{journaltitle}{Reports on Mathematical Physics}
  \field{month}{04}
  \field{year}{1976}
  \field{urlday}{17}
  \field{urlmonth}{03}
  \field{urlyear}{2008}
\endentry

\entry{audenaert_comparisons_2014}{article}{}
  \name{author}{1}{}{%
    {{}%
     {Audenaert}{A.}%
     {Koenraad M.~R.}{K.~M.~R.}%
     {}{}%
     {}{}}%
  }
  \strng{namehash}{AKMR1}
  \strng{fullhash}{AKMR1}
  \verb{eprint}
  \verb 1207.1197
  \endverb
  \field{issue}{1\&2}
  \field{pages}{31\bibrangedash 38}
  \field{title}{Comparisons between Quantum State Distinguishability Measures}
  \verb{url}
  \verb http://www.rintonpress.com/journals/qicabstracts/qicabstracts14-12.html
  \endverb
  \field{volume}{14}
  \field{journaltitle}{Quantum Information and Computation}
  \field{eprinttype}{arxiv}
  \field{eprintclass}{quant-ph}
  \field{day}{01}
  \field{month}{01}
  \field{year}{2014}
  \warn{\item Can't use 'eprinttype' + 'archiveprefix'}
\endentry

\entry{fuchs_cryptographic_1999}{article}{useprefix}
  \name{author}{2}{}{%
    {{}%
     {Fuchs}{F.}%
     {C.A.}{C.}%
     {}{}%
     {}{}}%
    {{}%
     {Graaf}{G.}%
     {J.}{J.}%
     {van~de}{v.~d.}%
     {}{}}%
  }
  \strng{namehash}{FCvdGJ1}
  \strng{fullhash}{FCvdGJ1}
  \verb{doi}
  \verb 10.1109/18.761271
  \endverb
  \verb{eprint}
  \verb quant-ph/9712042
  \endverb
  \field{issn}{0018-9448}
  \field{number}{4}
  \field{pages}{1216\bibrangedash 1227}
  \field{title}{Cryptographic distinguishability measures for
  quantum-mechanical states}
  \field{volume}{45}
  \field{journaltitle}{IEEE Transactions on Information Theory}
  \field{eprinttype}{arxiv}
  \field{year}{1999}
  \warn{\item Can't use 'eprinttype' + 'archiveprefix'}
\endentry

\entry{powers_free_1970}{article}{}
  \name{author}{2}{}{%
    {{}%
     {Powers}{P.}%
     {Robert~T.}{R.~T.}%
     {}{}%
     {}{}}%
    {{}%
     {St\o{}rmer}{S.}%
     {Erling}{E.}%
     {}{}%
     {}{}}%
  }
  \strng{namehash}{PRTSE1}
  \strng{fullhash}{PRTSE1}
  \verb{doi}
  \verb 10.1007/BF01645492
  \endverb
  \field{issn}{0010-3616, 1432-0916}
  \field{number}{1}
  \field{pages}{1\bibrangedash 33}
  \field{shortjournal}{Comm. Math. Phys.}
  \field{title}{Free states of the canonical anticommutation relations}
  \verb{url}
  \verb http://projecteuclid.org/euclid.cmp/1103842028
  \endverb
  \field{volume}{16}
  \field{journaltitle}{Communications in Mathematical Physics}
  \field{year}{1970}
  \field{urlday}{15}
  \field{urlmonth}{05}
  \field{urlyear}{2016}
\endentry

\entry{berta_entanglement-assisted_2014}{article}{}
  \name{author}{3}{}{%
    {{}%
     {Berta}{B.}%
     {Mario}{M.}%
     {}{}%
     {}{}}%
    {{}%
     {Coles}{C.}%
     {Patrick~J.}{P.~J.}%
     {}{}%
     {}{}}%
    {{}%
     {Wehner}{W.}%
     {Stephanie}{S.}%
     {}{}%
     {}{}}%
  }
  \strng{namehash}{BMCPJWS1}
  \strng{fullhash}{BMCPJWS1}
  \verb{doi}
  \verb 10.1103/PhysRevA.90.062127
  \endverb
  \verb{eprint}
  \verb 1302.5902
  \endverb
  \field{number}{6}
  \field{pages}{062127}
  \field{shortjournal}{Phys. Rev. A}
  \field{title}{Entanglement-assisted guessing of complementary measurement
  outcomes}
  \verb{url}
  \verb http://link.aps.org/doi/10.1103/PhysRevA.90.062127
  \endverb
  \field{volume}{90}
  \field{journaltitle}{Physical Review A}
  \field{eprinttype}{arxiv}
  \field{eprintclass}{quant-ph}
  \field{day}{22}
  \field{month}{12}
  \field{year}{2014}
  \field{urlday}{11}
  \field{urlmonth}{06}
  \field{urlyear}{2016}
  \warn{\item Can't use 'eprinttype' + 'archiveprefix'}
\endentry

\entry{buhrman_possibility_2008}{article}{}
  \name{author}{5}{}{%
    {{}%
     {Buhrman}{B.}%
     {Harry}{H.}%
     {}{}%
     {}{}}%
    {{}%
     {Christandl}{C.}%
     {Matthias}{M.}%
     {}{}%
     {}{}}%
    {{}%
     {Hayden}{H.}%
     {Patrick}{P.}%
     {}{}%
     {}{}}%
    {{}%
     {Lo}{L.}%
     {Hoi-Kwong}{H.-K.}%
     {}{}%
     {}{}}%
    {{}%
     {Wehner}{W.}%
     {Stephanie}{S.}%
     {}{}%
     {}{}}%
  }
  \strng{namehash}{BHCMHPLHKWS1}
  \strng{fullhash}{BHCMHPLHKWS1}
  \verb{doi}
  \verb 10.1103/PhysRevA.78.022316
  \endverb
  \verb{eprint}
  \verb 1302.5902
  \endverb
  \field{number}{2}
  \field{pages}{022316}
  \field{shortjournal}{Phys. Rev. A}
  \field{title}{Possibility, impossibility, and cheat sensitivity of
  quantum-bit string commitment}
  \verb{url}
  \verb http://link.aps.org/doi/10.1103/PhysRevA.78.022316
  \endverb
  \field{volume}{78}
  \field{journaltitle}{Physical Review A}
  \field{eprinttype}{arxiv}
  \field{eprintclass}{quant-ph}
  \field{day}{11}
  \field{month}{08}
  \field{year}{2008}
  \field{urlday}{11}
  \field{urlmonth}{06}
  \field{urlyear}{2016}
  \warn{\item Can't use 'eprinttype' + 'archiveprefix'}
\endentry

\entry{belavkin_optimal_1975}{article}{}
  \name{author}{1}{}{%
    {{}%
     {Belavkin}{B.}%
     {V.~P.}{V.~P.}%
     {}{}%
     {}{}}%
  }
  \strng{namehash}{BVP1}
  \strng{fullhash}{BVP1}
  \verb{doi}
  \verb 10.1080/17442507508833114
  \endverb
  \field{issn}{0090-9491}
  \field{number}{1}
  \field{pages}{315}
  \field{title}{Optimal multiple quantum statistical hypothesis testing}
  \verb{url}
  \verb http://www.informaworld.com/10.1080/17442507508833114
  \endverb
  \field{volume}{1}
  \field{journaltitle}{Stochastics}
  \field{year}{1975}
  \field{urlday}{04}
  \field{urlmonth}{03}
  \field{urlyear}{2010}
\endentry

\entry{hausladen_`pretty_1994}{article}{}
  \name{author}{2}{}{%
    {{}%
     {Hausladen}{H.}%
     {Paul}{P.}%
     {}{}%
     {}{}}%
    {{}%
     {Wootters}{W.}%
     {William~K.}{W.~K.}%
     {}{}%
     {}{}}%
  }
  \strng{namehash}{HPWWK1}
  \strng{fullhash}{HPWWK1}
  \verb{doi}
  \verb 10.1080/09500349414552221
  \endverb
  \field{issn}{0950-0340}
  \field{number}{12}
  \field{pages}{2385}
  \field{title}{A `Pretty Good' Measurement for Distinguishing Quantum States}
  \verb{url}
  \verb http://www.informaworld.com/10.1080/09500349414552221
  \endverb
  \field{volume}{41}
  \field{journaltitle}{Journal of Modern Optics}
  \field{year}{1994}
  \field{urlday}{17}
  \field{urlmonth}{03}
  \field{urlyear}{2008}
\endentry

\entry{dalla_pozza_optimality_2015}{article}{}
  \name{author}{2}{}{%
    {{}%
     {Dalla~Pozza}{D.~P.}%
     {Nicola}{N.}%
     {}{}%
     {}{}}%
    {{}%
     {Pierobon}{P.}%
     {Gianfranco}{G.}%
     {}{}%
     {}{}}%
  }
  \strng{namehash}{DPNPG1}
  \strng{fullhash}{DPNPG1}
  \verb{doi}
  \verb 10.1103/PhysRevA.91.042334
  \endverb
  \verb{eprint}
  \verb 1504.04908
  \endverb
  \field{number}{4}
  \field{pages}{042334}
  \field{shortjournal}{Phys. Rev. A}
  \field{title}{Optimality of square-root measurements in quantum state
  discrimination}
  \verb{url}
  \verb http://link.aps.org/doi/10.1103/PhysRevA.91.042334
  \endverb
  \field{volume}{91}
  \field{journaltitle}{Physical Review A}
  \field{eprinttype}{arxiv}
  \field{eprintclass}{quant-ph}
  \field{day}{27}
  \field{month}{04}
  \field{year}{2015}
  \field{urlday}{08}
  \field{urlmonth}{07}
  \field{urlyear}{2016}
  \warn{\item Can't use 'eprinttype' + 'archiveprefix'}
\endentry

\entry{helstrom_bayes-cost_1982}{article}{}
  \name{author}{1}{}{%
    {{}%
     {Helstrom}{H.}%
     {C.}{C.}%
     {}{}%
     {}{}}%
  }
  \keyw{Authentication, Computational Intelligence Society, Cost function,
  Cryptography, Electrons, Equations, Feedback, Hilbert space, Quantum
  detection, Sufficient conditions, Testing}
  \strng{namehash}{HC1}
  \strng{fullhash}{HC1}
  \field{abstract}{%
  An iterative procedure is described for reducing the Bayes cost in decisions
  among quantum hypotheses by minimizing the average cost in binary decisions
  between all possible pairs of hypotheses: the resulting decision strategy is
  a projection-valued measure and yields an upper bound to the minimum
  attainable Bayes cost. From it is derived an algorithm for finding the
  optimum measurement states for choosing among linearly independent pure
  states with minimum probability of error. The method is also applied to
  decisions among unimodal coherent quantum signals in thermal noise.%
  }
  \verb{doi}
  \verb 10.1109/TIT.1982.1056470
  \endverb
  \field{issn}{0018-9448}
  \field{number}{2}
  \field{pages}{359\bibrangedash 366}
  \field{title}{Bayes-cost reduction algorithm in quantum hypothesis testing
  {{(Corresp.)}}}
  \field{volume}{28}
  \verb{file}
  \verb IEEE Xplore Abstract Record:/Users/raban/Library/Application Support/Zo
  \verb tero/Profiles/i6ug0eiv.default/zotero/storage/5BMFSIAZ/articleDetails.h
  \verb tml:text/html;IEEE Xplore Full Text PDF:/Users/raban/Library/Applicatio
  \verb n Support/Zotero/Profiles/i6ug0eiv.default/zotero/storage/4P3KJ8U3/Hels
  \verb trom - 1982 - Bayes-cost reduction algorithm in quantum hypothes.pdf:ap
  \verb plication/pdf
  \endverb
  \field{journaltitle}{{IEEE} Transactions on Information Theory}
  \field{month}{03}
  \field{year}{1982}
\endentry

\entry{tomamichel_leftover_2011}{article}{}
  \name{author}{4}{}{%
    {{}%
     {Tomamichel}{T.}%
     {M.}{M.}%
     {}{}%
     {}{}}%
    {{}%
     {Schaffner}{S.}%
     {C.}{C.}%
     {}{}%
     {}{}}%
    {{}%
     {Smith}{S.}%
     {A.}{A.}%
     {}{}%
     {}{}}%
    {{}%
     {Renner}{R.}%
     {R.}{R.}%
     {}{}%
     {}{}}%
  }
  \strng{namehash}{TMSCSARR1}
  \strng{fullhash}{TMSCSARR1}
  \verb{doi}
  \verb 10.1109/TIT.2011.2158473
  \endverb
  \verb{eprint}
  \verb 1002.2436
  \endverb
  \field{issn}{0018-9448}
  \field{number}{8}
  \field{pages}{5524\bibrangedash 5535}
  \field{title}{Leftover Hashing Against Quantum Side Information}
  \field{volume}{57}
  \field{langid}{english}
  \field{journaltitle}{IEEE Transactions on Information Theory}
  \field{eprinttype}{arxiv}
  \field{eprintclass}{quant-ph}
  \field{month}{08}
  \field{year}{2011}
  \warn{\item Can't use 'eprinttype' + 'archiveprefix'}
\endentry

\entry{renes_one-shot_2012}{article}{}
  \name{author}{2}{}{%
    {{}%
     {Renes}{R.}%
     {J.~M}{J.~M}%
     {}{}%
     {}{}}%
    {{}%
     {Renner}{R.}%
     {R.}{R.}%
     {}{}%
     {}{}}%
  }
  \strng{namehash}{RJMRR1}
  \strng{fullhash}{RJMRR1}
  \verb{doi}
  \verb 10.1109/TIT.2011.2177589
  \endverb
  \verb{eprint}
  \verb 1008.0452
  \endverb
  \field{issn}{0018-9448}
  \field{number}{3}
  \field{pages}{1985\bibrangedash 1991}
  \field{title}{One-Shot Classical Data Compression With Quantum Side
  Information and the Distillation of Common Randomness or Secret Keys}
  \field{volume}{58}
  \field{langid}{english}
  \field{journaltitle}{IEEE Transactions on Information Theory}
  \field{eprinttype}{arxiv}
  \field{eprintclass}{quant-ph}
  \field{month}{03}
  \field{year}{2012}
  \warn{\item Can't use 'eprinttype' + 'archiveprefix'}
\endentry

\entry{golden65}{article}{}
  \name{author}{1}{}{%
    {{}%
     {Golden}{G.}%
     {Sidney}{S.}%
     {}{}%
     {}{}}%
  }
  \list{publisher}{1}{%
    {American Physical Society}%
  }
  \strng{namehash}{GS1}
  \strng{fullhash}{GS1}
  \verb{doi}
  \verb 10.1103/PhysRev.137.B1127
  \endverb
  \field{issue}{4B}
  \field{pages}{B1127\bibrangedash B1128}
  \field{title}{Lower Bounds for the {H}elmholtz Function}
  \verb{url}
  \verb http://link.aps.org/doi/10.1103/PhysRev.137.B1127
  \endverb
  \field{volume}{137}
  \field{journaltitle}{Physical Review}
  \field{year}{1965}
  \warn{\item Invalid format of field 'month'}
\endentry

\entry{thompson65}{article}{}
  \name{author}{1}{}{%
    {{}%
     {Thompson}{T.}%
     {Colin~J.}{C.~J.}%
     {}{}%
     {}{}}%
  }
  \strng{namehash}{TCJ1}
  \strng{fullhash}{TCJ1}
  \verb{doi}
  \verb 10.1063/1.1704727
  \endverb
  \field{number}{11}
  \field{pages}{1812\bibrangedash 1813}
  \field{title}{Inequality with Applications in Statistical Mechanics}
  \field{volume}{6}
  \field{journaltitle}{Journal of Mathematical Physics}
  \field{year}{1965}
\endentry

\entry{HP93}{article}{}
  \name{author}{2}{}{%
    {{}%
     {Hiai}{H.}%
     {Fumio}{F.}%
     {}{}%
     {}{}}%
    {{}%
     {Petz}{P.}%
     {D\'enes}{D.}%
     {}{}%
     {}{}}%
  }
  \strng{namehash}{HFPD1}
  \strng{fullhash}{HFPD1}
  \verb{doi}
  \verb 10.1016/0024-3795(93)90029-N
  \endverb
  \field{pages}{153 \bibrangedash  185}
  \field{title}{The {G}olden-{T}hompson trace inequality is complemented}
  \field{volume}{181}
  \field{journaltitle}{Linear Algebra and its Applications}
  \field{year}{1993}
\endentry

\entry{SBT15}{article}{}
  \name{author}{3}{}{%
    {{}%
     {Sutter}{S.}%
     {David}{D.}%
     {}{}%
     {}{}}%
    {{}%
     {Berta}{B.}%
     {Mario}{M.}%
     {}{}%
     {}{}}%
    {{}%
     {Tomamichel}{T.}%
     {Marco}{M.}%
     {}{}%
     {}{}}%
  }
  \strng{namehash}{SDBMTM1}
  \strng{fullhash}{SDBMTM1}
  \verb{eprint}
  \verb 1604.03023
  \endverb
  \field{note}{to appear in Communications in Mathematical Physics}
  \field{title}{Multivariate trace inequalities}
  \field{eprinttype}{arxiv}
  \field{eprintclass}{quant-ph}
  \field{month}{04}
  \field{year}{2016}
  \warn{\item Can't use 'eprinttype' + 'archiveprefix'}
\endentry

\entry{Lieb73}{article}{}
  \name{author}{1}{}{%
    {{}%
     {Lieb}{L.}%
     {Elliott~H}{E.~H.}%
     {}{}%
     {}{}}%
  }
  \strng{namehash}{LEH1}
  \strng{fullhash}{LEH1}
  \verb{doi}
  \verb 10.1016/0001-8708(73)90011-X
  \endverb
  \field{number}{3}
  \field{pages}{267 \bibrangedash  288}
  \field{title}{Convex trace functions and the {W}igner-{Y}anase-{D}yson
  conjecture}
  \field{volume}{11}
  \field{journaltitle}{Advances in Mathematics}
  \field{year}{1973}
\endentry

\entry{frank_monotonicity_2013}{article}{}
  \name{author}{2}{}{%
    {{}%
     {Frank}{F.}%
     {Rupert~L.}{R.~L.}%
     {}{}%
     {}{}}%
    {{}%
     {Lieb}{L.}%
     {Elliott~H.}{E.~H.}%
     {}{}%
     {}{}}%
  }
  \strng{namehash}{FRLLEH1}
  \strng{fullhash}{FRLLEH1}
  \verb{doi}
  \verb 10.1063/1.4838835
  \endverb
  \verb{eprint}
  \verb 1306.5358
  \endverb
  \field{issn}{0022-2488, 1089-7658}
  \field{number}{12}
  \field{pages}{122201}
  \field{title}{Monotonicity of a relative {{R{\'e}nyi}} entropy}
  \verb{url}
  \verb http://scitation.aip.org/content/aip/journal/jmp/54/12/10.1063/1.483883
  \verb 5
  \endverb
  \field{volume}{54}
  \field{journaltitle}{Journal of Mathematical Physics}
  \field{eprinttype}{arxiv}
  \field{eprintclass}{math-ph}
  \field{day}{01}
  \field{month}{12}
  \field{year}{2013}
  \field{urlday}{12}
  \field{urlmonth}{03}
  \field{urlyear}{2015}
  \warn{\item Can't use 'eprinttype' + 'archiveprefix'}
\endentry

\entry{dupuis_decoupling_2014}{article}{}
  \name{author}{3}{}{%
    {{}%
     {Dupuis}{D.}%
     {F.}{F.}%
     {}{}%
     {}{}}%
    {{}%
     {Szehr}{S.}%
     {O.}{O.}%
     {}{}%
     {}{}}%
    {{}%
     {Tomamichel}{T.}%
     {M.}{M.}%
     {}{}%
     {}{}}%
  }
  \strng{namehash}{DFSOTM1}
  \strng{fullhash}{DFSOTM1}
  \verb{doi}
  \verb 10.1109/TIT.2013.2295330
  \endverb
  \verb{eprint}
  \verb 1207.0067
  \endverb
  \field{issn}{0018-9448}
  \field{number}{3}
  \field{pages}{1562\bibrangedash 1572}
  \field{title}{A {{Decoupling Approach}} to {{Classical Data Transmission Over
  Quantum Channels}}}
  \field{volume}{60}
  \field{journaltitle}{IEEE Transactions on Information Theory}
  \field{eprinttype}{arxiv}
  \field{eprintclass}{quant-ph}
  \field{month}{03}
  \field{year}{2014}
  \warn{\item Can't use 'eprinttype' + 'archiveprefix'}
\endentry

\entry{sebastiani_derivatives_1996-1}{article}{}
  \name{author}{1}{}{%
    {{}%
     {Sebastiani}{S.}%
     {P.}{P.}%
     {}{}%
     {}{}}%
  }
  \strng{namehash}{SP1}
  \strng{fullhash}{SP1}
  \verb{doi}
  \verb 10.1137/S089547989528274X
  \endverb
  \field{issn}{0895-4798}
  \field{number}{3}
  \field{pages}{640\bibrangedash 648}
  \field{shortjournal}{SIAM. J. Matrix Anal. \& Appl.}
  \field{title}{On the Derivatives of Matrix Powers}
  \verb{url}
  \verb http://epubs.siam.org/doi/abs/10.1137/S089547989528274X
  \endverb
  \field{volume}{17}
  \field{journaltitle}{SIAM Journal on Matrix Analysis and Applications}
  \field{day}{01}
  \field{month}{07}
  \field{year}{1996}
  \field{urlday}{28}
  \field{urlmonth}{07}
  \field{urlyear}{2016}
\endentry

\entry{Lowner1934}{article}{}
  \name{author}{1}{}{%
    {{}%
     {L{\"o}wner}{L.}%
     {K.}{K.}%
     {}{}%
     {}{}}%
  }
  \strng{namehash}{LK1}
  \strng{fullhash}{LK1}
  \field{pages}{177\bibrangedash 216}
  \field{title}{{{\"U}}ber monotone {{M}}atrixfunktionen}
  \verb{url}
  \verb http://eudml.org/doc/168495
  \endverb
  \field{volume}{38}
  \field{journaltitle}{Mathematische Zeitschrift}
  \field{year}{1934}
\endentry

\entry{watrous_semidefinite_2009}{article}{}
  \name{author}{1}{}{%
    {{}%
     {Watrous}{W.}%
     {John}{J.}%
     {}{}%
     {}{}}%
  }
  \strng{namehash}{WJ1}
  \strng{fullhash}{WJ1}
  \verb{doi}
  \verb 10.4086/toc.2009.v005a011
  \endverb
  \verb{eprint}
  \verb 0901.4709
  \endverb
  \field{pages}{217\bibrangedash 238}
  \field{title}{Semidefinite {{Programs}} for {{Completely Bounded Norms}}}
  \field{volume}{5}
  \field{journaltitle}{Theory of Computing}
  \field{eprinttype}{arxiv}
  \field{eprintclass}{quant-ph}
  \field{year}{2009}
  \warn{\item Can't use 'eprinttype' + 'archiveprefix'}
\endentry

\lossort
\endlossort

 \blx@bblend
 \endgroup
 \csnumgdef{blx@labelnumber@\the\c@refsection}{0}}
\makeatother

\usepackage{titlesec}

\titleformat*{\section}{\bfseries}
\titleformat*{\subsection}{\normalsize\bfseries}
\titleformat*{\subsubsection}{\bfseries}
\titleformat*{\paragraph}{\large\bfseries}
\titleformat*{\subparagraph}{\large\bfseries}

\titlespacing\section{0pt}{12pt plus 4pt minus 2pt}{2pt plus 2pt minus 2pt}

\usepackage[bookmarks,colorlinks,breaklinks]{hyperref}  
\definecolor{dullmagenta}{rgb}{0.4,0,0.4}   
\definecolor{darkblue}{rgb}{0,0,0.4}
\hypersetup{linkcolor=red,citecolor=blue,filecolor=dullmagenta,urlcolor=darkblue} 

\newcommand{\norm}[1]{\left\lVert#1\right\rVert}

\RequirePackage[charter, greekuppercase=italicized]{mathdesign}
\RequirePackage{beramono}
\RequirePackage{berasans}

\usepackage{amsbsy,verbatim,graphicx}

\usepackage{braket}
\newcommand{\ketbra}[1]{|#1\rangle\langle #1|}
\def\tr{{\rm tr}}
\def\Re{{\rm Re}}
\def\cl{{\text{cl}}}
\def\ker{{\text{ker}}}

\def\cD{\mathcal D}
\def\cE{\mathcal E}

\renewcommand{\rho}{\varrho}
\renewcommand{\phi}{\varphi}
\def\id{\mathbbm 1}
\def\fpg{F_{{\rm pg}}}

\newtheorem{theorem}{Theorem}[section]
\newtheorem{corollary}[theorem]{Corollary}
\newtheorem{prop}[theorem]{Proposition}
\newtheorem{lemma}[theorem]{Lemma}

\theoremstyle{definition}

\newtheorem{rmk}[theorem]{Remark}

\usepackage{xfrac}

\usepackage{setspace}
\usepackage{titling}
\newcommand \myabstract[2][.8]{%
  \renewcommand\maketitlehookd{%
    \centering
    \begin{minipage}{#1\textwidth}
      {\begin{spacing}{1.0} \small #2\end{spacing}}
    \end{minipage}}}

\DeclareFontFamily{U}{mathx}{\hyphenchar\font45}
\DeclareFontShape{U}{mathx}{m}{n}{<-> mathx10}{}
\DeclareSymbolFont{mathx}{U}{mathx}{m}{n}
\DeclareMathAccent{\widebar}{0}{mathx}{"73}

\begin{document}

\title{\large {\bf Pretty good measures in quantum information theory}}

\author{
{\normalsize 
Raban Iten, Joseph M.\ Renes, and David Sutter}\\
\emph{\small 
Institute for Theoretical Physics, ETH Z\"urich, Switzerland}
}

\date{}
\myabstract{
Quantum generalizations of R\'enyi's entropies are a useful tool to describe a variety of operational tasks in quantum information processing. 
Two families of such generalizations turn out to be particularly useful: the Petz quantum R\'enyi divergence~$\widebar{D}_{\alpha}$ and the minimal quantum R\'enyi divergence~$\widetilde{D}_{\alpha}$. 
In this paper, we prove a reverse Araki-Lieb-Thirring inequality that implies a new relation between these two families of divergences, namely that $\alpha \widebar{D}_{\alpha}(\rho \| \sigma) \leq  \widetilde{D}_{\alpha}(\rho \| \sigma)$ for $\alpha \in [0,1]$ and where $\rho$ and $\sigma$ are density operators. 
This bound suggests defining a ``pretty good fidelity'', whose relation to the usual fidelity implies the known relations between the optimal and pretty good measurement as well as the optimal and pretty good singlet fraction.
We also find a new necessary and sufficient condition for optimality of the pretty good measurement and singlet fraction.

\vspace{4mm}

}

\maketitle

\section{Introduction}

As with their classical counterparts, quantum generalizations of R\'enyi entropies and divergences are powerful tools in information theory.
Two families of quantum  R\'enyi divergences have proven particularly useful, finding application to achievability, strong converses, and refined asymptotic analysis of a variety of coding and hypothesis testing problems (for a recent overview, see~\cite{tomamichel_quantum_2016}): the \emph{Petz quantum R\'enyi divergence}~\cite{petz_quasi-entropies_1986} and the \emph{minimal quantum R\'enyi divergence}~\cite{muller-lennert_quantum_2013,wilde_strong_2014} (also known as \emph{sandwiched quantum R\'enyi divergence}). 
A natural and important issue is the relation between these two families. 
In this work we prove a novel two-sided bound that relates the two families and discuss its implications. 	

For two non-negative operators $\rho\neq 0$ and $\sigma$ and $\alpha \in (0,1)\cup (1,\infty)$, the Petz quantum R\'enyi divergence is defined as
\begin{equation} \label{eq:petz_entropy}
\widebar{D}_{\alpha}(\rho \| \sigma):= \begin{cases}
\frac{1}{\alpha-1} \log \frac{1}{\tr \rho} \widebar{Q}_{\alpha}(\rho \| \sigma) &\text{if $ \sigma \gg \rho \lor \alpha<1 $}\\
\infty &\text{otherwise}\, ,
\end{cases}
\end{equation}
where $\widebar{Q}_{\alpha}(\rho||\sigma):=   \tr \rho^\alpha \sigma^{1-\alpha}$ and we use the common convention that  $-\log 0=\infty$. Moreover, negative matrix powers should be considered as generalized inverses. The notation $\sigma \gg \rho$ denotes that the kernel of $\sigma$ is a subset of the kernel of $\rho$. The minimal quantum R\'enyi divergence on the other hand is defined by 
\begin{equation}  \label{eq:min_entropy} 
\widetilde{D}_{\alpha}(\rho \| \sigma):= \begin{cases}
\frac{1}{\alpha-1} \log  \frac{1}{\tr \rho}  \widetilde{Q}_{\alpha}(\rho \| \sigma) &\text{if $ \sigma \gg \rho \lor \alpha<1 $}\\
\infty &\text{otherwise} \, ,
\end{cases}
\end{equation}
where $\widetilde{Q}_{\alpha}(\rho||\sigma):=\tr \left(\sigma^{\frac{1-\alpha}{2\alpha}}\rho\sigma^{\frac{1-\alpha}{2\alpha}}\right)^\alpha$. Moreover, we define $D_{0}$, $D_{1}$ and $D_{\infty}$  as limits of $D_{\alpha}$ for $\alpha \rightarrow 0$, $\alpha \rightarrow 1$ and $\alpha \rightarrow \infty$, respectively. 
Throughout this paper we use the convention that statements without either bar or tilde symbols are true for both cases.

The Araki-Lieb-Thirring (ALT) inequality~\cite{lieb_inequalities_1976,araki_inequality_1990} implies that the Petz divergence is larger than or equal to the minimal divergence, i.e., $\widebar{D}_{\alpha}(\rho \| \sigma) \geq \widetilde{D}_{\alpha}(\rho \| \sigma)$. 
But what remains unanswered is how much bigger than the minimal divergence the Petz divergence can be. 
We settle this question for $\alpha \leq1$ by showing that $\widebar{D}_{\alpha}(\rho \| \sigma) \leq \frac{1}{\alpha} \widetilde{D}_{\alpha}(\rho \| \sigma)$ if $\rho$ and $\sigma$ are normalized. 
This result follows from a new reversed ALT inequality. (We refer to Theorem~\ref{thm:Inv_ALT} and Corollary~\ref{cor:Upper_bound} for precise statements.)

This result has several applications. In Section~\ref{sec:pgf}, we define the  ``pretty good fidelity'' as $\fpg(\rho ,\sigma):=\tr \sqrt{\rho}\sqrt{\sigma}$. The result above then implies that the pretty good fidelity is indeed pretty good in that $\fpg \leq F\leq \sqrt{\fpg}$, where $F$ denotes the usual fidelity defined by $F(\rho,\sigma):=\tr (\sqrt{\rho} \sigma \sqrt{\rho})^{\nicefrac{1}{2}}$.  Analogous bounds are also known between the pretty good guessing probability and the optimal guessing probability~\cite{barnum_reversing_2002} as well as between the pretty good and the optimal achievable singlet fraction~\cite{dupuis_entanglement_2015}.\footnote{Note that ``singlet'' refers to a maximally entangled state (and not necessarily to the maximally entangled two-qubit state)~\cite{dupuis_entanglement_2015}. }  We show that both of these relations follow by the inequality relating the pretty good fidelity and the fidelity. We thus present a unified picture of the relationship between pretty good quantities and their optimal versions.
Additionally, we show that equality conditions for the ALT inequality lead to a new necessary and sufficient condition on the optimality of both pretty good measurement and singlet fraction.

\vspace{3mm}

 In this paper we consider finite-dimensional Hilbert spaces only, though most of our results can be extended to separable Hilbert spaces. 
We label Hilbert spaces with capital letters $A$, $B$, etc.\ and denote their dimension by $|A|$, $|B|$, etc..
The set of density operators on $A$, i.e., non-negative operators $\rho_A$ with $\tr\rho_A=1$, is denoted $\cD(A)$. 
We shall also make use of the convention $\frac{1}{0}=\infty$.
The \emph{Schatten} $p$\emph{-norm} of any linear operator $L$ is given by 
\begin{align}
\norm{L}_p:=\big(\tr |L|^p\big)^{\frac{1}{p}} \quad \text{for} \quad p\geq 1 \ ,
\end{align}
where $|L|:=\sqrt{L^{*} L}$. 
We may extend this definition to all $p>0$, but note that $\norm{L}_p$ is not a norm for $p \in (0,1)$ since it does not satisfy the triangle inequality.
In the limit $p\to \infty$ we recover the \emph{operator norm} and for $p=1$ we obtain the \emph{trace norm}. 
Schatten norms are functions of the singular values and thus unitarily invariant. Moreover, they satisfy $\|L\|_p = \|L^{*}\|_p$ and $\|L\|_{2p}^2 = \|LL^{*}\|_p= \|L^{*} L\|_p$.

\section{Results}

\subsection{Reverse ALT inequality}

The ALT inequality states that for any non-negative operators $A$ and $B$, $q\geq 0$ and $r \in [0,1]$, 
\begin{align}
\label{eq:alt}
\tr\, (B^{\frac{r}{2}} A^r B^{\frac{r}{2}})^q \leq \tr \, (B^{\frac{1}{2}} A B^{\frac{1}{2}})^{rq},
\end{align}
and the inequality holds in the opposite direction for $r \geq 1$~\cite{lieb_inequalities_1976,araki_inequality_1990}. Our main result is a reversed version of the ALT inequality. 

\begin{theorem}[Reverse ALT inequality]  \label{thm:Inv_ALT}
Let $A$ and $B$ be non-negative operators and $q > 0$. Then, for $r \in (0,1]$ and $a,b\in (0,\infty]$ such that $\frac{1}{2rq}=\frac{1}{2q}+\frac{1}{a}+\frac{1}{b}$, we have
\begin{align}
\label{eq:Inv_ALT}
\tr\,\big(B^{\frac{1}{2}}A B^{\frac{1}{2}}\big)^{rq} \leq \Big(\tr\,\big(B^{\frac{r}{2}}A^r B^{\frac{r}{2}}\big)^q\Big)^r \norm{A^{\frac{1-r}{2}}}_a^{2rq}\norm{B^{\frac{1-r}{2}} }_b^{2rq} \, .
\end{align}
Meanwhile, for $r \in [1,\infty)$ and  $a,b\in (0,\infty]$ such that $\frac{1}{2q}=\frac{1}{2rq}+\frac{1}{a}+\frac{1}{b}$, we have 
\begin{align}
\label{eq:Inv_ALT_r_bigger_one}
\tr\,\big(B^{\frac{1}{2}}A B^{\frac{1}{2}}\big)^{rq} \geq \Big(\tr\,\big(B^{\frac{r}{2}}A^r B^{\frac{r}{2}}\big)^q\Big)^r \norm{A^{\frac{r-1}{2}}}_a^{-2rq}\norm{B^{\frac{r-1}{2}} }_b^{-2rq} \, .
\end{align}
\end{theorem}

\begin{proof}
For $r=1$ the statement is trivial.
Let  $r \in (0,1)$ and $q > 0$. Recall the generalized H\"older inequality for matrices (see e.g.,~\cite[Exercise~IV.2.7]{bhatia_matrix_1997} for a proof): For $s$, $s_1,\dots,s_n$ positive real numbers and $\{A_k\}_{k=1}^n$ a collection of square matrices, it holds that 
\begin{align}
\label{eq:hoelder}
\left\|\prod_{k=1}^n A_k\right\|_s\leq\prod_{k=1}^n\norm{A_k}_{s_k} \qquad \text{for}\qquad \sum_{k=1}^n\frac1{s_k}=\frac1s \, .
\end{align}
Furthermore, we can rewrite the trace-terms in~\eqref{eq:Inv_ALT} as Schatten (quasi-)norms  
\begin{align}
\tr\,\big(B^{\frac{1}{2}}A B^{\frac{1}{2}}\big)^{rq}=\norm{ B^{\frac{1}{2}}  A^{\frac{1}{2}}}_{2rq}^{2rq} \qquad \text{and} \qquad
\tr\,\big(B^{\frac{r}{2}}A^r B^{\frac{r}{2}}\big)^{q}=\norm{ B^{\frac{r}{2}}  A^{\frac{r}{2}}}_{2q}^{2q} \, .
\end{align}
Inequality~\eqref{eq:Inv_ALT} then follows by an application of the generalized H\"older inequality with $n=3$. Choosing $s=2rq$, and $s_1=b$, $s_2=2q$, and $s_3=a$ for some $a,b \in (0,\infty]$ with $\frac{1}{2rq}=\frac{1}{2q}+\frac{1}{a}+\frac{1}{b}$, we find
\begin{align}
\tr\,\big(B^{\frac{1}{2}}A B^{\frac{1}{2}}\big)^{rq}=\norm{  B^{\frac{1-r}{2}}  B^{\frac{r}{2}}  A^{\frac{r}{2}} A^{\frac{1-r}{2}} }_{2rq}^{2rq} \leq \norm{B^{\frac{1-r}{2}} }_b^{2rq} \norm{B^{\frac{r}{2}}  A^{\frac{r}{2}}}_{2q}^{2rq} \norm{A^{\frac{1-r}{2}} }_a^{2rq} \, .
\end{align}
Inequality~\eqref{eq:Inv_ALT_r_bigger_one} now follows  from~\eqref{eq:Inv_ALT} by substituting  $A \rightarrow A^r$, $B \rightarrow B^r$, $r \rightarrow \frac{1}{r}$, and $q\rightarrow qr$. 
 \end{proof}

 \begin{rmk} \label{rmk:Inv_ALT}
Another reverse ALT inequality was given in~\cite{audenaert_araki-lieb-thirring_2008}, where it was shown that for $r \in (0,1)$ and $q > 0$ we have
\begin{align}
\label{eq:adenauer}
\tr(B^{\frac{1}{2}}A B^{\frac{1}{2}})^{rq} \leq \big(\tr(B^{\frac{r}{2}}A^r B^{\frac{r}{2}})^q\big)^r \big(\tr \, A^{rq}\norm{B}_{\infty}^{rq}\big)^{1-r} \, ,
\end{align}
while for $r>1$ the inequality holds in the opposite direction. We recover these inequalities as a corollary of Theorem~\ref{thm:Inv_ALT} by setting $b=\infty$ and $a=\frac{2rq}{1-r}$ in~\eqref{eq:Inv_ALT}, and $b=\infty$ and $a=\frac{2rq}{r-1}$ in~\eqref{eq:Inv_ALT_r_bigger_one}. We note that there also exists a reverse ALT inequality in terms of matrix means (see e.g.~\cite{ando94}) that however is different to Theorem~\ref{thm:Inv_ALT}. 
\end{rmk}

\subsection{Relation between the Petz and the minimal divergence}

It is known that the minimal  quantum R\'enyi divergence provides a lower bound for all other quantum R\'enyi divergences satisfying a small number of axiomatic properties (see e.g.,~\cite[\S4.2.2]{tomamichel_quantum_2016} for a precise statement). Hence, in particular, we have $\widetilde{D}_{\alpha}(\rho \| \sigma)\leq \widebar{D}_{\alpha}(\rho \| \sigma)$ for all $\alpha \in [0,\infty]$.\footnote{Alternatively, this follows directly from the ALT inequality.} Theorem~\ref{thm:Inv_ALT} leads to reversed relations between these two divergences. In the case where $\alpha \in [0,1]$, we find a particularly  useful relation of a simple form.  
\begin{corollary}  \label{cor:Upper_bound}
Let $\rho \neq 0$ and $\sigma$ be two non-negative operators and $\alpha \in [0,1]$. Then
\begin{align}
\label{eq:Upper_bound}
\alpha \widebar{D}_{\alpha}(\rho||\sigma)+(1-\alpha) (\log \tr \rho - \log \tr \sigma) \leq \widetilde{D}_{\alpha}(\rho||\sigma) \leq \widebar {D}_{\alpha}(\rho||\sigma) \, .
\end{align}
\end{corollary}
\begin{proof}
The second inequality is a direct consequence of the ALT inequality. 
It thus remains to show the first inequality. We note that it suffices to consider the case $\alpha \in (0,1)$, as $\alpha \in \{0,1\}$ then follows by continuity. By definition, we can reformulate the first inequality of~\eqref{eq:Upper_bound} as 
\begin{align} \label{eq:reformulated_thm}
\widetilde{Q}_{\alpha}(\rho||\sigma) \leq \widebar{Q}_{\alpha}(\rho||\sigma)^{\alpha}  (\tr \rho )^{\alpha(1-\alpha)}  (\tr \sigma )^{(1-\alpha)^2} \, . 
\end{align}
This follows from Theorem~\ref{thm:Inv_ALT} with $q=1$, $r=\alpha$, $A=\rho$,  $B=\sigma^{\frac{1-\alpha}{\alpha}}$, $a=\frac{2}{1-\alpha}$, and $b=\frac{2\alpha}{(1-\alpha)^2}$.
 \end{proof}

There is a well known equality condition for the ALT inequality, which leads to an equality condition for the second inequality of~\eqref{eq:Upper_bound}.
 \begin{lemma} \label{lem:eq_cond_ALT}
For $\alpha \in (0,1)$, we have  $\widetilde{D}_{\alpha}(\rho||\sigma) = \widebar {D}_{\alpha}(\rho||\sigma)$ if and only if $\rho$ and $\sigma$ commute. 
\end{lemma}
\begin{proof}
To see this, note that for  $r \in (1,\infty)$ and $rq \geq 1$, we have equality in the ALT inequality~\eqref{eq:alt} if and only if $A$ and $B$ commute. Equality for commuting states is obvious; for the other direction, note that we can rewrite~\eqref{eq:alt} using the substitution $rq = q'$ as
 \begin{align}
 \label{eq:alt_norm_version}
\norm{ (B^{\frac{r}{2}} A^r B^{\frac{r}{2}})^{\frac{1}{r}}}_{q'}\ \geq \norm{ (B^{\frac{1}{2}} A B^{\frac{1}{2}})}_{q'}\, .
\end{align}
 Equality in the inequality~\eqref{eq:alt_norm_version} for some  $r \in (1,\infty)$ (and noting that we have also equality for $r=1$) implies that the function $r \mapsto \| (B^{\frac{r}{2}} A^r B^{\frac{r}{2}})^{\frac{1}{r}}\|_{q'} $ is not strictly increasing.  Therefore, by \cite[Theorem 2.1]{hiai_equality_1994}, it follows\footnote{Here we use our assumption that  $q' \geq 1$, since in this case $\norm{\cdot}_{q'}$ is a strictly increasing norm.} that $[A,B]=0$. Let  $\rho,\sigma$ be non negative. Setting $r=\nicefrac{1}{\alpha},q=\alpha$ and $A=\rho^{\alpha}$, $B=\sigma^{1-\alpha}$ in~\eqref{eq:alt}, we conclude that for $\alpha \in (0,1)$ we have that $\widetilde{D}_{\alpha}(\rho||\sigma) = \widebar {D}_{\alpha}(\rho||\sigma)$ if and only if $[\rho,\sigma]=0$.
 \end{proof}

For density operators $\rho$ and $\sigma$ the first inequality of Corollary~\ref{cor:Upper_bound} simplifies to 
\begin{align}
\alpha \widebar{D}_{\alpha}(\rho \| \sigma) \leq \widetilde{D}_{\alpha}(\rho \| \sigma) \quad \text{for} \quad \alpha \in [0,1]\, .
\end{align}
This bound is simpler than an alternative bound given in~\cite{mosonyi_coding_2015}, which is based on the earlier reversed ALT inequality in \eqref{eq:adenauer} and states that 
$ \alpha \widebar{D}_{\alpha}(\rho \| \sigma) -\log \tr \rho^{\alpha} +(\alpha-1) \log \norm{\sigma}_{\infty} \leq \widetilde{D}_{\alpha}(\rho \| \sigma)$ for density operators $\rho$ and $\sigma$.

\subsection{Relations between quantum conditional R\'enyi entropies}
Divergences can be used to define conditional entropies. For any density operator $\rho_{AB}$ on $A\otimes B$ we define the \emph{quantum conditional R\'enyi entropy} of $A$ given $B$ as
\begin{align} \label{eq:cond_entropies}
H^{\downarrow}_{\alpha}(A|B)_{\rho}:=-D_{\alpha}(\rho_{AB}\| \id_A \otimes \rho_B) \quad \text{and} \quad 
H^{\uparrow}_{\alpha}(A|B)_{\rho}:=\sup \limits_{\sigma_B \in \cD(B) }-D_{\alpha}(\rho_{AB}\| \id_A \otimes \sigma_B) \, . 
\end{align}
Note that the special cases $\alpha\in \{0,1,\infty\}$ are defined by taking the limits inside the supremum.\footnote{We are following the notation in~\cite{tomamichel_quantum_2016}. Note that $H_{\text{min}}(A|B)_{\rho| \rho}=\widetilde{H}^{\downarrow}_{\infty}(A|B)_{\rho}$, $H_{\text{min}}(A|B)_{\rho}=\widetilde{H}^{\uparrow}_{\infty}(A|B)_{\rho}$ and $H_{\text{max}}(A|B)_{\rho}=\widetilde{H}^{\uparrow}_{\frac{1}{2}}(A|B)_{\rho}$ are also often used notations.
} 
We call the set of all conditional entropies with $\alpha \in (0,1)$ ``max-like'' and those with $\alpha \in (1,\infty)$ ``min-like'', owing to the fact that under small changes to the state the entropies in either class are approximately equal~\cite{renner_smooth_2004,tomamichel_fully_2009}. 
Moreover, min- and max-like entropies are related by some interesting duality relations, which are summarized in the following lemma.

\begin{lemma} [Duality relations~\cite{tomamichel_fully_2009, tomamichel_relating_2014, muller-lennert_quantum_2013, beigi_sandwiched_2013,konig_operational_2009, berta_single-shot_2008}] \label{lem_duality}

Let $\rho_{ABC}$ be a pure state on $A \otimes B \otimes C$. Then
\begin{align}
\widebar{H}^{\downarrow}_{\alpha}(A|B)_{\rho}+\widebar{H}^{\downarrow}_{\beta}(A|C)_{\rho} =0 &\quad \text{when} \quad \alpha+\beta=2 \, \text{ for } \alpha, \beta \in [0,2] \, \quad \text{and} \\
\widetilde{H}^{\uparrow}_{\alpha}(A|B)_{\rho}+\widetilde{H}^{\uparrow}_{\beta}(A|C)_{\rho} =0 &\quad \text{when} \quad \frac{1}{\alpha}+\frac{1}{\beta}=2 \, \text{ for } \alpha, \beta \in [\frac{1}{2},\infty] \, \quad \text{and} \\
\widebar{H}^{\uparrow}_{\alpha}(A|B)_{\rho}+\widetilde{H}^{\downarrow}_{\beta}(A|C)_{\rho} =0 &\quad \text{when} \quad \alpha \beta =1 \, \text{ for } \alpha, \beta \in [0,\infty] \, ,
\end{align} 
where we use the convention that $\frac{1}{\infty}=0$ and $\infty \cdot 0=1\,$.
\end{lemma}

\subsubsection{Relations between max-like entropies}
As a direct consequence of Corollary~\ref{cor:Upper_bound}, we find the following relation between conditional max-like entropies.

\begin{corollary} \label{cor:entropies}
For $\alpha \in [0,1]$ and $\rho_{AB}\in \cD(A\otimes B)\,$, we have that
\begin{align}
&\widebar{H}^{\downarrow}_{\alpha}(A|B)_{\rho} \leq \widetilde{H}^{\downarrow}_{\alpha}(A|B)_{\rho}\leq \alpha \widebar{H}^{\downarrow}_{\alpha}(A|B)_{\rho}+(1-\alpha)\log |A|  \label{eq:Bounds_entropies1} \qquad \text{and}  \\
&\widebar{H}^{\uparrow}_{\alpha} (A|B)_{\rho}\leq \widetilde{H}^{\uparrow}_{\alpha}(A|B)_{\rho}\leq \alpha \widebar{H}^{\uparrow}_{\alpha}(A|B)_{\rho}+(1-\alpha)\log |A| \label{eq:Bounds_entropies2} \, .
\end{align}
\end{corollary}

We can further improve the upper bounds in~\eqref{eq:Bounds_entropies1} and~\eqref{eq:Bounds_entropies2} by removing the second term if $\rho_{AB}$ has a special structure consisting of a quantum and a  classical part that is handled coherently.

\begin{prop} \label{ccq_states}
Let $\ket{\rho}_{XX'BB'}=\sum_x \sqrt{p_x}\ket{x}_X\ket{x}_{X'}\ket{\xi_x}_{BB'}$ be a pure state on $X\otimes X'\otimes B\otimes B'$, where $X'\simeq X$, $p_x \in [0,1]$ with $\sum_x p_x=1\,$, and the pure states $\ket{\xi_x}_{BB'}$ are arbitrary. 
Then
\begin{align}
&\widetilde{H}^{\downarrow}_{\alpha}(X|X'B)_{\rho}\leq \alpha \widebar{H}^{\downarrow}_{\alpha}(X|X'B)_{\rho} \quad \text{ for  $\alpha \in [0,1]$}  \label{eq:Bounds_ccq1}  \qquad \text{and}\\
& \widetilde{H}^{\uparrow}_{\alpha}(X|X'B)_{\rho}\leq \alpha \widebar{H}^{\uparrow}_{\alpha}(X|X'B)_{\rho}  \quad \text{ for  $\alpha \in [\tfrac{1}{2},1]$}  \label{eq:Bounds_ccq2} \, .
\end{align}
\end{prop}
States $\rho_{XX'B}$ are sometimes called ``classically coherent'' as the classical information is treated coherently, i.e.\ fully quantum-mechanically. 
\begin{proof} [Proof of Proposition~\ref{ccq_states}]
It is known that $\widetilde{D}_1=\widebar{D}_1$ (see for example~\cite{tomamichel_quantum_2016}), and hence the claim is trivial in the case $\alpha=1$. Using~\eqref{eq:cond_entropies} as well as~\eqref{eq:petz_entropy} and~\eqref{eq:min_entropy} , one can see that it suffices to show that
\begin{align}
&\widetilde{Q}_{\alpha}(\rho_{XX'B}\| \id_X\otimes \rho_{X'B})\leq \widebar{Q}_{\alpha}(\rho_{XX'B}\| \id_X\otimes \rho_{X'B})^{\alpha} \quad \text{ for  $\alpha \in (0,1)$}  \label{eq:Q_formulation_cq1} \qquad \text{and}\\
&\widetilde{Q}_{\alpha}(\rho_{XX'B}\| \id_X\otimes \sigma_{X'B})\leq \widebar{Q}_{\alpha}(\rho_{XX'B}\| \id_X\otimes \sigma_{X'B})^{\alpha} \quad \text{ for  $\alpha \in [\tfrac{1}{2},1)$} \, , \label{eq:Q_formulation_cq2}
\end{align}
 for all density operators $\sigma_{X'B}$ (the case $\alpha=0$ then follows by continuity).

The marginal state $\rho_{X'B}$ appearing in \eqref{eq:Q_formulation_cq1} is a classical quantum (cq) state by assumption. Importantly, by the monotonicity of the R\'enyi divergence, we need only prove~\eqref{eq:Q_formulation_cq2} for cq states $\sigma_{X'B}$ in order to show~\eqref{eq:Bounds_ccq2}. Indeed, by Lemma~\ref{lem:DPI_dephasing} of Appendix~\ref{app:dephasing},  the supremum arising in equation~\eqref{eq:Bounds_ccq2} can be taken only over cq states.

Now define the unitary $U_{XX'}:=\sum_{x',x} \ket{x-x'}\bra{x}_X\otimes \ketbra{x'}_{X'}$, where arithmetic inside the ket is taken modulo $|X|$, and observe that $U_{XX'}\otimes \id_B$ leaves the state $\id_X \otimes \sigma_{X'B}$ invariant (here we use the assumption that  $\sigma_{X'B}$ is a cq state). 
Hence, by unitary invariance of $Q_{\alpha}$, we find
\begin{align}
Q_{\alpha}(  \rho_{XX'B}\|\id_X \otimes \sigma_{X'B})
&=Q_{\alpha}\big((U_{XX'}\otimes \id_B )\rho_{XX'B} (U_{XX'}^{*}\otimes \id_B) \| \id_X \otimes \sigma_{X'B}\big) \\
&= Q_{\alpha}\big(\ketbra{0}_X\otimes \sum_{x,x'}\sqrt{p_xp_{x'}}\ket{x}\bra{x'}_{X'}\otimes \tr_{B'}\ket{\xi_x}\bra{\xi_{x'}}_{BB'} \|\id_X \otimes \sigma_{X'B}\big)  \\
&=  Q_{\alpha}\Big(\sum_{x,x'}\sqrt{p_xp_{x'}}\ket{x}\bra{x'}_{X'}\otimes \tr_{B'}\ket{\xi_x}\bra{\xi_{x'}}_{BB'} \|\sigma_{X'B}\Big) \label{eq:Q_ccq_rewritten}\, ,
\end{align}
where we used the multiplicity of the trace under tensor products in the last equality. The claim now follows by a direct application of Corollary~\ref{cor:Upper_bound} (or more precisely of~\eqref{eq:reformulated_thm} applied to density operators):
\begin{align}
\widetilde{Q}_{\alpha}(  \rho_{XX'B}\|\id_X \otimes \sigma_{X'B})
&=  \widetilde{Q}_{\alpha}\Big(\sum_{x,x'}\sqrt{p_xp_{x'}}\ket{x}\bra{x'}_{X'} \otimes \tr_{B'}\ket{\xi_x}\bra{\xi_{x'}}_{BB'} \|\sigma_{X'B}\Big) \\
&\leq \widebar{Q}_{\alpha}\Big(\sum_{x,x'}\sqrt{p_xp_{x'}}\ket{x}\bra{x'}_{X'} \otimes \tr_{B'}\ket{\xi_x}\bra{\xi_{x'}}_{BB'} \|\sigma_{X'B}\Big)^{\alpha}\\
&=\widebar{Q}_{\alpha}(  \rho_{XX'B}\|\id_X \otimes \sigma_{X'B})^{\alpha} \,.
\end{align}
This shows inequality~\eqref{eq:Q_formulation_cq2} for cq states $\sigma_{X'B}$, and hence~\eqref{eq:Bounds_ccq2}. Moreover, we recover inequality~\eqref{eq:Q_formulation_cq1} by setting $\sigma_{X'B}=\rho_{X'B}$.
 \end{proof}

\subsubsection{Relations between min-like entropies}

We can use duality relations for conditional entropies (see Lemma~\ref{lem_duality}) and Corollary~\ref{cor:entropies} to derive new bounds for conditional min-like entropies.

\begin{lemma} \label{lem:bounds_cond_entr}
For $\alpha \in [1,2]$ and $\rho_{AB}\in \cD(A\otimes B)\,$, we have that\footnote{We use again the convention that $\frac{1}{0}=\infty\,$.}
\begin{align}
&\widetilde{H}^{\downarrow}_{\alpha}(A|B)_{\rho} \leq \alpha \widetilde{H}^{\uparrow}_{\frac{1}{2-\alpha}}(A|B)_{\rho}+(\alpha-1) \log |A| \label{eq:bounds_cond_entr1} \qquad \text{and}\\
&\widebar{H}^{\downarrow}_{\alpha}(A|B)_{\rho} \leq \frac{1}{2-\alpha} \left( \widebar{H}^{\uparrow}_{ \frac{1}{2-\alpha}}(A|B)_{\rho}+(\alpha-1) \log |A|\right). \label{eq:bounds_cond_entr2}
\end{align}
\end{lemma}

\begin{proof}
Let $\tau_{ABC}$ be a purification of $\rho_{AB}$ on $A \otimes B\otimes C$, i.e., $\tau_{ABC}$ is a pure state with $\tr_C \tau_{ABC}=\rho_{AB}$. 
Then, we find
\begin{align}
\widetilde{H}^{\downarrow}_{\alpha}(A|B)_{\tau}
=-\widebar{H}^{\uparrow}_{\frac{1}{\alpha}}(A|C)_{\tau} 
\leq- \alpha \widetilde{H}^{\uparrow}_{\frac{1}{\alpha}}(A|C)_{\tau} +(\alpha-1) \log |A|
= \alpha \widetilde{H}^{\uparrow}_{\frac{1}{2-\alpha}}(A|B)_{\tau} +(\alpha-1) \log |A| \, ,
\end{align}
where we used  Corollary~\ref{cor:entropies} for the inequality and duality relations in the first and third equality. 
Similarly, we find
\begin{align}
\widebar{H}^{\downarrow}_{\alpha}(A|B)_{\tau}
&=-\widebar{H}^{\downarrow}_{2-\alpha}(A|C)_{\tau} \\
&\leq \frac{1}{2-\alpha} \left(-\widetilde{H}^{\downarrow}_{2-\alpha}(A|C)_{\tau} +(\alpha-1) \log |A|\right)\\
&=  \frac{1}{2-\alpha} \left(\widebar{H}^{\uparrow}_{\frac{1}{2-\alpha}}(A|B)_{\tau} +(\alpha-1) \log |A|\right) \, ,
\end{align}
where we again used Corollary~\ref{cor:entropies} for the inequality and duality relations in the first and third equality. 
\end{proof}

\begin{corollary} \label{cor:entropy_cq}
Let  $\alpha \in [1,2]$ and $\rho_{XB}$ be a cq state on $X \otimes B$, i.e., $\rho_{XB}=\sum_x p_x \ketbra{x}_X \otimes (\rho_x)_B$ where $(\rho_x)_B$ are density operators and $p_x \in [0,1]\,$, such that $\sum_x p_x=1\,$. Then
\begin{align}
&\widetilde{H}^{\downarrow}_{\alpha}(X|B)_{\rho} \leq \alpha \widetilde{H}^{\uparrow}_{\frac{1}{2-\alpha}}(X|B)_{\rho}  \label{eq:cond_entr_bound_cq1} \qquad \text{and} \\
&\widebar{H}^{\downarrow}_{\alpha}(X|B)_{\rho} \leq \frac{1}{2-\alpha}  \widebar{H}^{\uparrow}_{ \frac{1}{2-\alpha}}(X|B)_{\rho}.\label{eq:cond_entr_bound_cq2}
\end{align}
\end{corollary}

\begin{proof}
The proof proceeds analogously to the proof of Lemma~\ref{lem:bounds_cond_entr}, but we can make use of the improved bounds given in Proposition~\ref{ccq_states}: Let $\ket{\tau}_{XX'BB'}=\sum_x \sqrt{p_x}\ket{x}_X\ket{x}_{X'}\ket{\xi_x}_{BB'}$ where $\ket{\xi_x}_{BB'}$ purifies $(\rho_x)_B$. 
The system $X' \otimes B'$ corresponds to the system $C$ in the proof of Lemma~\ref{lem:bounds_cond_entr} and the state on $X \otimes X' \otimes B'$, i.e., $\tau_{XX'B'}$, is a classical-coherent state as required for Proposition~\ref{ccq_states} (note that the role of $B$ and $B'$ are interchanged here and in the statement of Proposition~\ref{ccq_states}). 
\end{proof}

We note that the special case $\alpha=2$ of the inequalities~\eqref{eq:bounds_cond_entr1}  and~\eqref{eq:cond_entr_bound_cq1} was already shown in~\cite{dupuis_entanglement_2015}.

\subsubsection{Equality condition for max-like  entropies} \label{sec:equality_cond_entropy}
In this section, we give a necessary and sufficient condition on a density operator $\rho_{AB}$, such that the entropies  $\widebar{H}^{\uparrow}_{\alpha}(A|B)_{\rho}$ and $\widetilde{H}^{\uparrow}_{\alpha}(A|B)_{\rho}$ are equal for $\alpha \in [\tfrac{1}{2},1)$. To derive the necessary condition, let $\alpha \in (0,1)$. In the proof of Lemma~1 of~\cite{tomamichel_relating_2014}, it is shown that the optimizer $\sigma^\star_B$ of $\widebar{H}^{\uparrow}_{\alpha}(A|B)_{\rho}=\sup_{\sigma_B \in \cD(B) }-\widebar{D}_{\alpha}(\rho_{AB}\| \id_A \otimes \sigma_B)$ is given by
 \begin{align} \label{def:tilde_sigma}
 \sigma^\star_B=\frac{\left(\tr_A \, \rho_{AB}^{\alpha}\right)^{\frac{1}{\alpha}}}{\tr\, \left(\tr_A{\rho_{AB}^{\alpha}}\right)^{\frac{1}{\alpha}}}\, .
 \end{align} 
 By the ALT inequality~\cite{lieb_inequalities_1976,araki_inequality_1990}, we then find that
  \begin{align}  \label{eq:ALT_for_entropies}
  \widebar{H}^{\uparrow}_{\alpha}(A|B)_{\rho}=-\widebar{D}_{\alpha}(\rho_{AB}||\id_A \otimes \sigma^\star_B) \leq \sup \limits_{\sigma_B \in \cD(B) }-\widetilde{D}_{\alpha}(\rho_{AB}\| \id_A \otimes \sigma_B)
 \, .
 \end{align} 

According to Lemma~\ref{lem:eq_cond_ALT}, a necessary condition for equality in~\eqref{eq:ALT_for_entropies} is that  $[\rho_{AB},\id_A \otimes \sigma^\star_B]=0$. Assume now that $\alpha \in   [\tfrac{1}{2},1)$. To show that this condition is also sufficient for equality in~\eqref{eq:ALT_for_entropies}, it suffices to show that the function $\sigma_B \mapsto  -\widetilde{D}_{\alpha}(\rho_{AB}||\id_A \otimes \sigma_B)$ or equivalently $\sigma_B \mapsto  \widetilde{Q}_{\alpha}(\rho_{AB}||\id_A \otimes \sigma_B)$ attains its global maximum at $\sigma_B=\sigma^\star_B$ if $[\rho_{AB},\id_A \otimes \sigma^\star_B]=0$. The proof of this fact is based on standard derivative techniques, albeit for matrices, and is given in Appendix~\ref{app:equality_condition_extrema}. The results are summarized in the following Lemma.
\begin{lemma}[Equality condition for entropies] \label{lem:equal_cond_entropies}
Let $\alpha \in [\tfrac{1}{2},1)\,$, $\rho_{AB}$ be a density operator and $\hat{\sigma}^\star_B:=\tr_A \, \rho_{AB}^{\alpha}$. Then, the following are equivalent
\begin{enumerate}
\item $\widebar{H}^{\uparrow}_{\alpha}(A|B)_{\rho}=\widetilde{H}^{\uparrow}_{\alpha}(A|B)_{\rho} $
\item $[\rho_{AB},\id_A \otimes \hat{\sigma}^\star_{B}]=0$ .
\end{enumerate}
\end{lemma}

\section{Pretty good fidelity and the quality of pretty good measures}
Our main results yield a unified framework relating pretty good measures often used  in quantum information to their optimal counterparts.

\subsection{Pretty good fidelity} \label{sec:pgf}
Let   $\rho$ and $\sigma$ be two density operators throughout this subsection.  We define the pretty good fidelity of $\rho$ and $\sigma$ by
\begin{align}
\fpg(\rho,\sigma):=\widebar{Q}_{\frac{1}{2}}(\rho,\sigma)=\tr\,\sqrt\rho\sqrt\sigma \, .
\end{align}
This quantity was called the ``quantum affinity'' in \cite{luo_informational_2004} and is nothing but the fidelity of the ``pretty good purification'' introduced in~\cite{winter_extrinsic_2004}: Letting $\ket{\Omega}_{AA'}=\sum_k \ket{k}_A \ket{k}_{A'}$,  the canonical purification with respect to $\ket{\Omega}_{AA'}$ of $\rho$ is $\ket{\Psi_\rho}_{AA'}=({\sqrt{\rho}_A \otimes \id_{A'}}) \ket{\Omega}_{AA'}$, and thus 
\begin{align}
\fpg(\rho,\sigma)=\braket{\Psi_\rho|\Psi_\sigma}_{AA'}.
\end{align}

Recall that the usual fidelity is given by
\begin{align}
F(\rho,\sigma):=\widetilde{Q}_{\frac{1}{2}}(\rho,\sigma)
=\norm{\sqrt{\rho} \sqrt{\sigma}}_1
=\max_{V_{A'}}\bra{\Psi_\rho}(\id_{A}\otimes V_{A'})\ket{\Psi_\sigma}_{AA'}, 
\end{align}
where the maximum is taken over all unitary operators $V_{A'}$ and the final equality follows from Uhlmann's theorem~\cite{uhlmann_transition_1976}. Therefore, it is clear that $\fpg(\rho,\sigma)\leq F(\rho,\sigma)$. This can also be seen from the ALT inequality directly (cf. Corollary~\ref{cor:Upper_bound} for $\alpha=\tfrac 12$), and therefore, by Lemma~\ref{lem:eq_cond_ALT}, we have that $\fpg(\rho,\sigma)=F(\rho,\sigma)$ if and only if $[\rho, \sigma]=0$. 
The reverse ALT inequality implies a bound in the opposite direction; a similar approach using the H\"older inequality is given in \cite{audenaert_comparisons_2014}. 
By choosing $\alpha=\nicefrac{1}{2}$, it follows from Corollary~\ref{cor:Upper_bound} that the fidelity is also upper bounded by the square root of the pretty good fidelity, i.e.,
\begin{align} \label{ineq_pgf}
\fpg(\rho,\sigma)\leq F(\rho,\sigma)\leq \sqrt{\fpg(\rho,\sigma)}\, .
\end{align}
Hence the pretty good fidelity is indeed pretty good. 

Recall that the \emph{trace distance} between two density operators $\rho$ and $\sigma$ is defined by $\delta(\rho,\sigma):=\frac{1}{2}\|\rho-\sigma\|_1$.
An important property of the fidelity is its relation to the trace distance~\cite{fuchs_cryptographic_1999}:
\begin{align} \label{eq:fid_tr_dist}
1-F(\rho,\sigma) \leq \delta(\rho,\sigma) \leq \sqrt{1-F(\rho,\sigma)^2} \, .
\end{align}
Indeed the pretty good fidelity satisfies the same relation, 
\begin{align} 
\label{eq:fvdg}
1-\fpg(\rho,\sigma) \leq \delta(\rho,\sigma) \leq \sqrt{1-\fpg(\rho,\sigma)^2} \, .
\end{align}
The upper bound follows immediately by combining the upper bound in \eqref{eq:fid_tr_dist} with the lower bound in~\eqref{ineq_pgf}. 
The lower bound was first shown in \cite{powers_free_1970} (see also \cite{audenaert_comparisons_2014}). 



\subsection{Relation to bounds for the pretty good measurement and singlet fraction} \label{sec:bounds_pg_measures}
In this section we show that together with entropy duality, the relation between fidelity and pretty good fidelity in \eqref{ineq_pgf} implies the known optimality bounds of the pretty good measurement and the pretty good singlet fraction. 
Let us first consider the optimal and pretty good singlet fraction. 
Define $R(A|B)_\rho$ to be the largest achievable overlap with the maximally entangled state one can obtain from $\rho_{AB}$ by applying a quantum channel on $B$. 
Formally, 
\begin{align} 
R(A|B)_{\rho}:=\text{max}_{\mathcal{E}_{B \rightarrow A'}} F(\ket{\Phi}\!\bra{\Phi}_{AA'}, (\id_A \otimes \mathcal{E}_{B \rightarrow A'})\rho_{AB})^2 \, ,
\end{align}
where $\ket{\Phi}_{AA'}=\frac{1}{\sqrt{|A|}}\sum_k \ket{k}_A \ket{k}_{A'}$ and the maximization is over all completely positive, trace-preserving maps ${\mathcal{E}_{B \rightarrow A'}}$. 
In~\cite{konig_operational_2009} it was shown that 
\begin{align} 
\widetilde{H}^{\uparrow}_{\infty}(A|B)_{\rho}=- \log |A|\, R(A|B)_{\rho} \, .
\end{align}
A ``pretty good'' map $\cE_{\text{pg}}$ was considered in \cite{berta_entanglement-assisted_2014}, and it was shown that  
\begin{align} 
\widetilde{H}^{\downarrow}_{2}(A|B)_{\rho}= -\log |A|\, R_{\text{pg}}(A|B)_{\rho} \, , 
\end{align}
where $R_{\text{pg}}(A|B)_\rho$ is the overlap obtained by using $\cE_{\text{pg}}$. 
Clearly $R_{\text{pg}}(A|B)_{\rho}\leq R(A|B)_{\rho}$, but the case $\alpha=2$ in~\eqref{eq:bounds_cond_entr1}, which comes from \eqref{ineq_pgf} via entropy duality, implies that we also have 
\begin{align} 
R_{\text{pg}}(A|B)_{\rho} \leq R(A|B)_{\rho}\leq \sqrt{R_{\text{pg}}(A|B)_{\rho}} \, .
\end{align}
This was also shown in~\cite{dupuis_entanglement_2015}. 
Note that in the special case where $\rho_{AB}$ has the form of a Choi state, i.e., $\tr_B \, \rho_{AB}=\frac{1}{|A|} \id_A$, this statement also follows from~\cite{barnum_reversing_2002}. 

Now let $\rho_{XB}=\sum_x p_x \ketbra{x}_X \otimes (\rho_x)_B$ be a cq state, and consider an observer with access to the system $B$ who would like to guess the variable $X$. 
Denote by $p_{\text{guess}}(X|B)$ the optimal guessing probability which can be achieved by performing a POVM on the system $B$. 
It was shown in~\cite{konig_operational_2009} that 
\begin{align} 
\widetilde{H}^{\uparrow}_{\infty}(X|B)_{\rho} = -\log p_{\text{guess}}(X|B) \,  . 
\end{align}
On the other hand, it is also known that~\cite{buhrman_possibility_2008}
\begin{align} 
\widetilde{H}^{\downarrow}_{2}(X|B)_{\rho} = -\log p^{\text{pg}}_{\text{guess}}(X|B) \, ,
\end{align}
where $p_{\text{guess}}^{\text{pg}}(X|B)$ denotes the guessing probability of the pretty good measurement introduced in~\cite{belavkin_optimal_1975, hausladen_`pretty_1994}. 
Clearly $p^{\text{pg}}_{\text{guess}} (X|B) \leq p_{\text{guess}}(X|B)$, but the case $\alpha=2$ in~\eqref{eq:cond_entr_bound_cq1}, which again comes from \eqref{ineq_pgf} via entropy duality, also implies that  
\begin{align} 
p^{\text{pg}}_{\text{guess}} (X|B)  \leq p_{\text{guess}} (X|B)  \leq   \sqrt{p^{\text{pg}}_{\text{guess}} (X|B)}\, .  
\end{align}
This was originally shown in~\cite{barnum_reversing_2002}.

\subsection{Optimality conditions for pretty good measures}
Our framework also yields a novel optimality condition for the pretty good measures. 
Supposing $\tau_{ABC}$ is a purification of $\rho_{AB}$, the duality relations for R\'enyi entropies (cf. Lemma~\ref{lem_duality}) imply
\begin{align} \label{eq:aquival_dual_picture}
\widetilde{H}_{2}^{\downarrow}(A|B)_{\tau}=\widetilde{H}_{\infty}^{\uparrow}(A|B)_{\tau} \quad \iff \quad \widebar{H}^{\uparrow}_{\nicefrac{1}{2}}(A|C)_{\tau}= \widetilde{H}^{\uparrow}_{\nicefrac{1}{2}}(A|C)_{\tau}\, .
\end{align}
Applying the equality condition for max-like conditional entropies, using Lemma~\ref{lem:equal_cond_entropies}, we find that the pretty good singlet fraction and pretty good measurement are optimal if and only if ${[\tau_{AC}, \id_A \otimes \hat{\sigma}^\star_{C}]}=0$, where $\hat{\sigma}^\star_{C}:=\tr_A \, \sqrt{ \tau_{AC}}$. 
Alternately, this specific equality condition ($\alpha=\nicefrac12$) can be established via weak duality of semidefinite programs, as described in Appendix~\ref{app:sdp}.

As a simple example of optimality of the pretty good singlet fraction, consider the case of a pure bipartite $\rho_{AB}$. 
Then every purification $\tau_{ABC}=\rho_{AB}\otimes \xi_C$ for some pure $\xi_C$.
Thus, $\tau_{AC}=\rho_A\otimes \xi_C$, and it follows immediately that the optimality condition is satisfied. 
Optimality also holds for arbitrary mixtures of pure states, i.e., for states of the form $\rho_{ABY}=\sum_y q_y \ketbra{\psi_y}_{AB}\otimes \ketbra {y}_Y$ with some arbitrary distribution $q_y$, provided both $B$ and $Y$ are used in the entanglement recovery operation. 
Here any purification takes the form $\ket{\tau}_{ABYY'}=\sum_y \sqrt{q_y}\ket{\psi_y}_{AB}\ket{y}_Y\ket{y}_{Y'}$. Hence, we have that $\tau_{AY'}=\sum_y q_y \tr_B \,\ketbra{\psi_y}_{AB} \otimes \ketbra{y}_{Y'}$, a state in which $Y'$ is classical, for which it is easy to see that the optimality condition holds. 

The optimality condition for the pretty good measurement can be simplified using the classical coherent nature of the state $\tau_{AC}$, which results in a condition formulated in terms of the Gram matrix. 
Suppose $\rho_{XB}=\sum_x p_x \ketbra{x}_X \otimes (\rho_x)_B$ describes the ensemble of mixed states $(\rho_x)_B$, for which a natural purification is given by 
\begin{align} \label{eq:purification_of_rho}
\ket{\tau}_{XX'BB'}=\sum_x \sqrt{p_x}\ket{x}_X\ket{x}_{X'}\ket{\xi_x}_{BB'}\, ,
\end{align}
where $\ket{\xi_x}_{BB'}$ denotes a purification of $(\rho_x)_B$.
Then we define the (generalized) Gram matrix $G$ 
\begin{align} \label{eq:G_def}
G_{X'B'}:=\sum_{x,x'}   \sqrt{p_xp_{x'}} \ket{x}\bra{x'}_{X'}  \otimes \tr_{B} \, \ket{\xi_{x}}\bra{\xi_{x'}}_{BB'} \, .
\end{align}
This definition reverts to the usual Gram matrix when the states $(\rho_x)_B$ are pure and system $B'$ is trivial. 
Observe that we are in the setting of Proposition~\ref{ccq_states}; using the unitary  $U_{XX'}$ introduced in its proof, we find that $\left(U_{XX'}\otimes \id_{B'}\right)\tau_{XX'B'}\left(U_{XX'}^{*}\otimes \id_{B'}\right)=\ketbra{0}_X\otimes G_{X'B'}$. 
Hence, $\sqrt{\tau_{XX'B'}}=\left(U_{XX'}^{*}\otimes \id_{B'}\right)(\ketbra{0}_X\otimes \sqrt{G_{X'B'}})\left(U_{XX'}\otimes \id_{B'}\right)$ and a further calculation shows that $\tr_X\sqrt{\tau_{XX'B'}}=\hat{\sigma}^\star_{X'B'}$, with 
\begin{align}
\hat{\sigma}^\star_{X'B'}:=\sum_x \ketbra{x}_{X'} \otimes \bra{x} \sqrt{G_{X'B'}}\ket{x}_{X'}.
\end{align}
Note that $[M,N]=0$ is equivalent to $[UMU^{*},UNU^{*}]=0$ for any square matrices $M,N$ and any unitary $U$. Therefore, we find that the equality condition $[\tau_{XX'B'}, \id_X \otimes \hat{\sigma}^\star_{X'B'}]=0$ is equivalent to $[\ketbra{0}_X \otimes G_{X'B'}, \id_X \otimes\hat{\sigma}^\star_{X'B'}]=[G_{X'B'},\hat{\sigma}^\star_{X'B'}]=0$. Thus we have shown the following result:
\begin{lemma}[Optimality condition for the pretty good measurement] \label{lem:opt_pgm}
The pretty good measurement is optimal for distinguishing states in the ensemble $\{p_x,\rho_x\}$ if and only if $[G_{X'B'},\hat{\sigma}^\star_{X'B'}]=0$.
\end{lemma}

In the case of distinguishing pure states, we recover Theorem~2 of~\cite{dalla_pozza_optimality_2015} (which was first shown in~\cite{helstrom_bayes-cost_1982}). 
To see this, observe that $B'$ is now trivial and $G_{X'}$ is the usual Gram matrix. 
Moreover, $\hat{\sigma}^\star_{X'}$ is now the diagonal of the square root of $G_{X'}$, and the commutation condition of Lemma~\ref{lem:opt_pgm} becomes $[G_{X'},\hat{\sigma}^\star_{X'}]=0$, which is equivalent to the condition in equation~(11) of~\cite{dalla_pozza_optimality_2015} (in the case of the pretty good measurement). 
Reformulating what it means for the Gram matrix $G_{X'}$ to commute with the diagonal matrix $\hat{\sigma}^\star_{X'}$ then leads to Theorem~3 of~\cite{dalla_pozza_optimality_2015}.

\section{Conclusions}

We have given a novel reverse ALT inequality (see Theorem~\ref{thm:Inv_ALT}) that answers the question of how much bigger the Petz quantum R\'enyi divergence can be compared to the minimal quantum R\'enyi divergence for $\alpha \leq 1$. More precisely, together with the standard ALT inequality it implies that $\alpha \widebar D_{\alpha}(\rho \| \sigma) \leq \widetilde D_{\alpha}(\rho \| \sigma) \leq \widebar D_{\alpha}(\rho \| \sigma)$ for $\alpha \leq 1$ and any density operators $\rho$ and $\sigma$.
This bound leads to an elegant unified framework of pretty good constructions in quantum information theory, and the ALT equality condition leads to a simple necessary and sufficient condition for their optimality. 
Previously it was observed that the min entropy $\widetilde H^{\uparrow}_\infty$ characterizes optimal measurement and singlet fraction, while $\widetilde H_2^\downarrow$ is the ``pretty good min entropy'' since it characterizes pretty good measurement and singlet fraction. 
On the other hand, we can think of $\widebar H_{\nicefrac{1}{2}}^\uparrow$ as the ``pretty good max entropy'' since it is based on the pretty good fidelity instead of the (usual) fidelity itself as in the max entropy $\widetilde H_{\nicefrac{1}{2}}^{\uparrow}$. 
Entropy duality then beautifully links the two, as the (pretty good) max entropy is dual to the (pretty good) min entropy, and the known optimality bounds can be seen to stem from the lower bound on the pretty good fidelity in \eqref{ineq_pgf}. 
Indeed, that such a unified picture might be possible was the original inspriation to look for a reverse ALT inequality of the form given in Theorem~\ref{thm:Inv_ALT}. 
It is also interesting to note that both the pretty good min and max entropies appear in achievability proofs of information processing tasks, the former in randomness extraction against quantum adversaries~\cite{tomamichel_leftover_2011} and the latter in the data compression with quantum side information~\cite{renes_one-shot_2012}.

For future work, it would be interesting to elaborate more on the novel reverse ALT inequality (see Theorem~\ref{thm:Inv_ALT}). It is know that the ALT inequality implies the Golden-Thompson (GT) inequality~\cite{golden65,thompson65} via the Lie-Trotter product formula. Reverse versions of the GT inequality are well-studied~\cite{HP93}. It would be thus interesting to see if Theorem~\ref{thm:Inv_ALT} can be related to the reverse GT inequality. Recent progress on proving multivariate trace inequalities~\cite{SBT15} (see also~\cite{Lieb73}) suggests the possibility of an $n$-matrix extension of the reversed ALT inequality.

\vspace{5mm}

\noindent \textbf{Acknowledgements.}
We thank Volkher Scholz and Marco Tomamichel for helpful conversations.
JMR and DS acknowledge support by the Swiss National Science Foundation (SNSF) via the National Centre of Competence in Research ``QSIT'' and by the European Commission via the project ``RAQUEL''.

\begin{appendices}

\section{Optimal marginals for classically coherent states} \label{app:dephasing}
This appendix details the argument that cq states are optimal in the conditional entropy expressions for classically coherent states. 
First we recall the data processing inequality (DPI), which states that for all completely positive, trace-preserving maps $\mathcal{E}$ and for all non-negative operators $\rho$ and $\sigma$, we have
\begin{align}
D(\rho \| \sigma)\geq D\big(\mathcal{E}(\rho) \|\mathcal{E}( \sigma)\big) \, .
\end{align} 
It was shown that $\widebar{D}_{\alpha}$ satisfies the DPI for $\alpha \in (0,1) \cup (1,2]$ in~\cite{petz_quasi-entropies_1986}, while \cite{frank_monotonicity_2013} (see also~\cite{beigi_sandwiched_2013}) shows that $\widetilde{D}_{\alpha}$  satisfies the DPI for $\alpha \in [\frac{1}{2}, \infty]$.
Following the approach taken in \cite[Lemma A.1]{dupuis_decoupling_2014} to establish a similar result for the smooth min entropy, we can show 
\begin{lemma} \label{lem:DPI_dephasing}
Let   $ \ket{\rho}_{XX'BB'}=\sum_x \sqrt{p_x}\ket{x}_X\ket{x}_{X'}\ket{\xi_x}_{BB'}$ be a pure state on  $X\otimes X' \otimes B \otimes B'$, where $p_x \in [0,1]$ with $\sum_x p_x=1\,$, and $X'\simeq X$. Then, for any density operator $\sigma_{X'B}\,$, we have that 
 \begin{align}
 \label{eq:classicalcoherentQ}
&Q_{\alpha}(\rho_{XX'B}\| \id_X \otimes \sigma_{X'B})\leq Q_{\alpha}(\rho_{XX'B}\| \id_X \otimes \sigma^{\cl}_{X'B}) \quad \text{for }  \alpha \in [\tfrac{1}{2},1)   \, ,
 \end{align}
 where $\sigma^{\cl}_{X'B}:=\sum_{x}\ketbra x_{X'}\otimes \bra{x}\sigma_{X'B}\ket{x}_{X'}$. 
\end{lemma}
\begin{proof} 
Let $P_{XX'}=\sum_x \ketbra{x}_{X} \otimes \ketbra{x}_{X'}$ and define the quantum channel $\mathcal E$ from $X \otimes X'$ to itself by $\mathcal{E}(\cdot):= P_{XX'}(\cdot)P_{XX'}+(\id_{XX'}-P_{XX'})(\cdot)(\id_{XX'}-P_{XX'})$. 
Since $P_{XX'}\ket{\Psi}_{XX'BB'}=\ket{\Psi}_{XX'BB'}$, $\mathcal{E}_{XX'}\otimes \mathcal I_B$ leaves the density operator $\rho_{XX'B}$ invariant.
By the DPI we then have, for  $\alpha \in [\tfrac{1}{2},1)$,
 \begin{align}
Q_{\alpha}(\rho_{XX'B}\| \id_X \otimes \sigma_{X'B})
&\leq Q_{\alpha}\big(\rho_{XX'B}\| \mathcal{E}_{XX'} \otimes \mathcal I_B(\id_X \otimes \sigma_{X'B})\big) \\
&=Q_{\alpha}\big(\rho_{XX'B}\| (P_{XX'} \otimes \id_B)(\id_X \otimes \sigma_{X'B})(P_{XX'} \otimes \id_B)\big)\,.
 \end{align}
In the second line we use the fact that $Q_\alpha$ is indifferent to parts of its second argument which are not contained in the support of its first argument. 
Observe that $(P_{XX'} \otimes \id_B)(\id_X \otimes \sigma_{X'B})({P_{XX'} \otimes \id_B})=\sum_x \ketbra x_{X} \otimes \ketbra{x}_{X'}\sigma_{X'B}\ketbra{x}_{X'}\leq \id_{X}  \otimes \sigma^{\cl}_{X'B}$. Inequality~\eqref{eq:classicalcoherentQ} now follows directly from the dominance property of $D_{\alpha}$ (see e.g.,~\cite{tomamichel_quantum_2016}), which states (in terms of $Q_{\alpha}$) that $Q_{\alpha}(\rho\|  \sigma) \leq Q_{\alpha}(\rho\|  \sigma')$ for any non-negative operators $\rho, \sigma, \sigma'$ with $\sigma \leq \sigma'$ .
\end{proof}

\section{Sufficient condition for equality of max-like entropies} \label{app:equality_condition_extrema}

In this Appendix, we show that, for $\alpha \in [\tfrac{1}{2},1)$, the function  $f_{\alpha}: \cD(B) \ni \sigma_B \mapsto  \widetilde{Q}_{\alpha}(\rho_{AB}||\id_A \otimes \sigma_B)$  attains its global maximum at $\sigma_B=\sigma^\star_B$  if $[\rho_{AB},\id_A \otimes \sigma^\star_B]=0$. We use the notation of Section~\ref{sec:equality_cond_entropy}. The following lemma is similar to Lemma~5.1 of~\cite{sebastiani_derivatives_1996-1}. 

\begin{lemma} \label{lem:derivative}
Let $I \subset \mathbb{R}$ be open and  $t_0 \in I$. Let $A(t)$ be a matrix whose entries are smooth functions of $t \in I$ and $A(t)>0$ for all $t \in I$. 
Further, let $B$ be a matrix such that $[B,A(t_0)]=0$. Then,
\begin{align} \label{eq:derivative_of_matrix_power}
\frac{d}{dt}\Bigr|_{\substack{t=t_0}} \tr \, BA(t)^r=r \, \tr \, B A(t_0)^{r-1} A'(t_0)  \quad \text{for} \quad r \in \mathbb{R} \, ,
\end{align}
where $A'(t_0):=\frac{d}{dt}\Bigr|_{\substack{t=t_0}} A(t)$\, .
\end{lemma}
\begin{proof}
Note that it is straightforward to adapt Theorem 3.5 of~\cite{sebastiani_derivatives_1996-1} to the complex case. Therefore, by setting $\alpha=0$ in the equation~(26) of~\cite{sebastiani_derivatives_1996-1}, we find that
\begin{align}
\frac{d}{dt}\Bigr|_{\substack{t=t_0}} \tr \,  BA(t)^r =r \,  \tr   \, B  A'(t_0) A(t_0)^{r-1} + r \,  \tr \,  B H_{0,r} A(t_0)^{r-1}  \, ,
\end{align}
where  $H_{0,r}$ is defined in equation~(27) of~\cite{sebastiani_derivatives_1996-1}. Since $[A(t_0),B]=0$, a short calculation shows that $\tr \, B H_{0,r} A(t_0)^{r-1}=0$.
\end{proof}
\begin{lemma} \label{lem:derivative_tilde_Q}
Set $I=(-\delta,\delta) \subset \mathbb{R}$ for some $\delta>0$ and let $A(t)$ be a matrix whose entries are smooth functions of $t \in I$ and $A(t)>0$ for all $t \in I$.  For $B$ a density operator such that $[B,A(0)]=0$, 
\begin{align}
\frac{d}{dt}\Bigr|_{\substack{t=0}} \widetilde{Q}_{\alpha}(B||A(t))=(1-\alpha) \, \Re \, \tr \, B^{\alpha} A(0)^{-\alpha} A'(0)  \quad \text{for} \quad \alpha \in (0,1) \, ,
\end{align}
where $A'(t_0):=\frac{d}{dt}\Bigr|_{\substack{t=t_0}} A(t)$ for $t_0 \in I$.
\end{lemma}
\begin{proof}
To simplify the notation, let us define $\beta:=\tfrac{1-\alpha}{2\alpha}$. We set $B_{\varepsilon}:=B+\varepsilon \id>0$ for some $\varepsilon>0$. Using Lemma~\ref{lem:derivative} (with $A=A(t)^{\beta}B_{\varepsilon}A(t)^{\beta}$ and $B=\id$), we find  
\begin{align}
\frac{d}{dt}\Bigr|_{\substack{t=t_0}}& \tr \left(A(t)^{\beta}B_{\varepsilon}A(t)^{\beta}\right)^{\alpha} \nonumber \\
&= \alpha \, \tr \, \left(A(t_0)^{\beta}B_{\varepsilon}A(t_0)^{\beta}\right)^{\alpha-1} \frac{d}{dt}\Bigr|_{\substack{t=t_0}} \left(A(t)^{\beta}B_{\varepsilon}A(t)^{\beta} \right) \\
\begin{split}
&= \alpha \,  \tr \, \left(A(t_0)^{\beta}B_{\varepsilon}A(t_0)^{\beta}\right)^{\alpha-1} \left( \frac{d}{dt}\Bigr|_{\substack{t=t_0}} A(t)^{\beta}B_{\varepsilon}A(t_0)^{\beta} +A(t_0)^{\beta} B_{\varepsilon} \frac{d}{dt}\Bigr|_{\substack{t=t_0}}A(t)^{\beta}  \right) \, .
\end{split}
\end{align} 
This can be simplified by noting that for any Hermitian matrix $H$ and any matrix $C$, 
\begin{align}
\tr  \, H(C+C^{*})
=\tr \, HC+\tr \, H C^{*} 
= \tr \, HC+\tr \, H^{*}  C^{*}
=\tr \, HC+\left(\tr \, HC \right)^{\ast} 
=2 \, \text{Re} \, \tr \, HC .
\end{align} 
Using this we obtain
\begin{align} \label{eq:derivative_Q_calculation}
\frac{d}{dt}\Bigr|_{\substack{t=t_0}} \widetilde{Q}_{\alpha}\big(B_{\varepsilon}||A(t)\big)
&=2 \alpha \, \Re \, \tr  \, \left(A(t_0)^{\beta}B_{\varepsilon}A(t_0)^{\beta}\right)^{\alpha-1} \frac{d}{dt}\Bigr|_{\substack{t=t_0}} A(t)^{\beta}B_{\varepsilon}A(t_0)^{\beta} \\
&=2 \alpha \, \Re \, \tr \, A(t_0)^{-\beta} \left(A(t_0)^{\beta}B_{\varepsilon} A(t_0)^{\beta}\right)^{\alpha} \frac{d}{dt}\Bigr|_{\substack{t=t_0}} A(t)^{\beta}\, .
\end{align} 
Taking the limit $\varepsilon \rightarrow 0$ yields
\begin{align}
\lim_{\varepsilon\to 0}\frac{d}{dt}\Bigr|_{\substack{t=t_0}} \widetilde{Q}_{\alpha}\big(B_{\varepsilon}||A(t)\big)
&=2 \alpha \, \Re \, \tr \, A(t_0)^{-\beta} \left(A(t_0)^{\beta}B A(t_0)^{\beta}\right)^{\alpha} \frac{d}{dt}\Bigr|_{\substack{t=t_0}} A(t)^{\beta}\, .
\end{align}
At $t_0=0$ the righthand side can be simplified by again making use of Lemma~\ref{lem:derivative} as well as $[A(0),B]=0$:
\begin{align}
\lim_{\varepsilon\to 0}\frac{d}{dt}\Bigr|_{\substack{t=0}} \widetilde{Q}_{\alpha}\big(B_{\varepsilon}||A(t)\big)
&=2 \alpha \, \Re \, \tr  \, B^{\alpha} A(0)^{\beta (2\alpha-1)}\frac{d}{dt}\Bigr|_{\substack{t=0}} A(t)^{\beta} \\
&=(1-\alpha) \, \Re  \,  \tr \,  B^{\alpha} A(0)^{-\alpha} A'(0) \, .
\end{align}

It remains to be shown that the limit can be interchanged with the derivative. 
This follows if we ensure that $\frac{d}{dt}\bigr|_{\substack{t=t_0}} \widetilde{Q}_{\alpha}(B_{\varepsilon}||A(t))$ converges uniformly in $t_0 \in [-\nicefrac{\delta}{2},\nicefrac{\delta}{2}]$ for $\varepsilon \rightarrow 0$. 
To show uniform convergence, it suffices to show  
\begin{align}
\lim_{\varepsilon \rightarrow 0} \sup_{t_0 \in  [-\delta/2,\delta/2]} \norm{A(t_0)^{-\beta} \left[ \left(A(t_0)^{\beta}B_{\varepsilon} A(t_0)^{\beta}\right)^{\alpha}-\left(A(t_0)^{\beta}B A(t_0)^{\beta}\right)^{\alpha} \right] \frac{d}{dt}\Bigr|_{\substack{t=t_0}} A(t)^{\beta}}_1 =0\, ,
\end{align}
where we used that $|\tr (M)|\leqslant \norm{M}_1$ for any square matrix $M$ (see, e.g.,~\cite[Exercise~IV~2.12]{bhatia_matrix_1997}). By the generalized H\"older inequality for matrices (see~\eqref{eq:hoelder}), we find that it is enough to show that
\begin{align}
\lim_{\varepsilon \rightarrow 0} \sup_{t_0 \in  [-\delta/2,\delta/2]} \norm{A(t_0)^{-\beta}}_{\infty} \norm{ \left(A(t_0)^{\beta}B_{\varepsilon} A(t_0)^{\beta}\right)^{\alpha}-\left(A(t_0)^{\beta}B A(t_0)^{\beta}\right)^{\alpha} }_1\norm{ \frac{d}{dt}\Bigr|_{\substack{t=t_0}} A(t)^{\beta}}_{\infty} =0\, .
\end{align}
Note that the infinity-norm terms are bounded on the compact interval $t_0 \in  [-\nicefrac{\delta}{2},\nicefrac{\delta}{2}]$, as $A(t)^{\beta}$ is continuously differentiable for $A(t)>0$.
Thus, we need only show  that
\begin{align}\label{eq:unif_conv_matrix_powers}
\lim_{\varepsilon \rightarrow 0} \sup_{t_0 \in  [-\nicefrac{\delta}{2},\nicefrac{\delta}{2}]}  \norm{ \left(A(t_0)^{\beta}B_{\varepsilon} A(t_0)^{\beta}\right)^{\alpha}-\left(A(t_0)^{\beta}B A(t_0)^{\beta}\right)^{\alpha} }_1=0\, .
\end{align}
Since $t \rightarrow t^{\alpha}$ is operator monotone for $\alpha \in [0,1]$ (L\"owner's theorem~\cite{Lowner1934}), the matrix inside the trace norm is positive, and hence~\eqref{eq:unif_conv_matrix_powers} is equivalent to
\begin{align}\label{eq:unif_conv_matrix_powers2}
\lim_{\varepsilon \rightarrow 0} \sup_{t_0 \in  [-\nicefrac{\delta}{2},\nicefrac{\delta}{2}]} \norm{ \left(A(t_0)^{\beta}B_{\varepsilon} A(t_0)^{\beta}\right)^{\alpha}}_1- \norm{\left(A(t_0)^{\beta}B A(t_0)^{\beta}\right)^{\alpha} }_1=0\, .
\end{align}
Note that $\varepsilon \mapsto \|(A(t_0)^{\beta}B_{\varepsilon} A(t_0)^{\beta})^{\alpha}\|_1$ is monotonically decreasing (again by  L\"owner's theorem). Then, by Dini's theorem,  it converges uniformly to  $\|(A(t_0)^{\beta}B A(t_0)^{\beta})^{\alpha}\|_1$, which proves~\eqref{eq:unif_conv_matrix_powers2}, and hence the desired uniformity of the convergence. 
\end{proof}
We are now ready to calculate the derivative of the function  $f_{\alpha}$ at $\sigma_B=\sigma^\star_B$.
\begin{lemma}
\label{lem:alphasufficiency}
Let $\alpha \in [\tfrac{1}{2},1)$ and $\rho_{AB} \in \cD(A\otimes B)$ be such that $[\rho_{AB}, \id_A \otimes \sigma^\star_B]=0$. Then the  function $f_{\alpha}:\cD(B) \ni \sigma_B\mapsto  \widetilde{Q}_{\alpha}(\rho_{AB}||\id_A \otimes \sigma_B)$ attains its global maximum at $\sigma^\star_B$ as defined in \eqref{def:tilde_sigma}. 
\end{lemma}
\begin{proof} 
First consider the case $\rho_{AB}>0$ for simplicity; we return to the rank-deficient case below. 
Since $(\rho,\sigma) \mapsto \widetilde{Q}_{\alpha}(\rho\| \sigma)$ is jointly concave~\cite{frank_monotonicity_2013, beigi_sandwiched_2013}, the function $ f_{\alpha} :  \mathcal{D}(B)\ni \sigma_B \mapsto {\widetilde Q_\alpha(\rho_{AB} \| \id_A \otimes \sigma_B)}$ is concave. As $ \mathcal{D}(B)$ is a convex set, it suffices to show that $f_\alpha$ has an extreme point at $\sigma^\star_B$ (which is then also a global maximum). 
Observe that $\sigma^\star_B>0$ by definition, and therefore all states $\sigma_B(t)$ along arbitrary paths of states through $\sigma_B(0)=\sigma^\star_B$ have full rank for all $t$ sufficiently close to zero.  
Thus, we may use Lemma~\ref{lem:derivative_tilde_Q} to compute the derivative along any such path and find  
\begin{align}
\frac{d}{dt}\Bigr|_{\substack{t=0}}  \widetilde{Q}_{\alpha}(\rho_{AB}||\id_A \otimes \sigma_B(t)) 
&= (1-\alpha) \, \Re \,  \tr \, \rho_{AB}^{\alpha} \left(\id_A \otimes \sigma^\star_B\right)^{-\alpha}  \left(\id_A \otimes \frac{d}{dt}\Bigr|_{\substack{t=0}} \sigma_{B}(t) \right)\\
&= (1-\alpha) \, \Re \, \tr  \left( \tr_A \left(\rho_{AB}^{\alpha}\right) (\sigma^\star_B)^{-\alpha}  \frac{d}{dt}\Bigr|_{\substack{t=0}} \sigma_{B}(t)  \right)  
=0 \, .
\end{align} 
Therefore $\sigma^\star_B$ is the optimizer in this case. 

For $\rho_{AB}$ not strictly positive, we can restrict the set of marginal states $\sigma_B$ to the support of $\sigma^\star_B$ and replay the above argument. 
To see this, first observe that the support of $\sigma^\star_B$ is the same as that of $\rho_{B}$. 
Furthermore, as noted in~\cite{muller-lennert_quantum_2013}, the DPI for $\widetilde{D}_{\alpha}$ implies that the maximum of  $f_{\alpha}$ is always attained at a density matrix $\sigma^{\star}_B$ satisfying $\sigma^{\star}_B\ll \rho_B$.
Therefore, we can restrict the domain of the function $f_{\alpha}$ to the set $\mathcal{P}(B):=\{\sigma_B \in \cD(B): \sigma_B \ll \sigma^\star_B\}$. 
Now observe that $\ker (\id_A\otimes\sigma^\star_B)\subseteq \ker (\rho_{AB})$.
For any $\ket{\psi}_B$ we have $\bra{\psi}\rho_B\ket{\psi}_B=\sum_k \bra{k}_A \bra{\psi} _B \rho_{AB}\ket{k}_A \ket{\psi}_B$. 
By positivity of $\rho_{AB}\geq 0$, each $\ket{\psi}_B\in \ker(\sigma^\star_B)=\ker(\rho_B)$ leads to a set of states $\ket{k}_A\otimes \ket{\psi}_B\in \ker(\rho_{AB})$.
This implies that projecting  $\rho_{AB}$ to the support of $\id_A\otimes \sigma^\star_B$ has no effect on $\widetilde Q_{\alpha}$. 
Hence, we can restrict all operators in the problem to this subspace, where again all states in $\mathcal P(B)$ sufficiently close to $\sigma^\star_B$ have full rank.
\end{proof}

\section{Optimality condition for pretty good measures via semidefinite programming}
\label{app:sdp}

Here we derive the optimality condition for pretty good measures via weak duality of semidefinite programs. 
In terms of fidelity and pretty good fidelity, the optimality condition in \eqref{eq:aquival_dual_picture} reads 
\begin{align}
\label{eq:fpgopt}
\fpg(\tau_{AC},\id_A\otimes \sigma^\star_C)=\sup_{\sigma\in \cD(C)} F(\tau_{AC},\id_A\otimes \sigma_C),
\end{align}
where $ \sigma^\star_C$ is as in \eqref{def:tilde_sigma} with $\alpha=\nicefrac 12$. 
Lemma~\ref{lem:eq_cond_ALT} implies that $[\tau_{AC},\id_A\otimes \sigma^\star_C]=0$ is necessary for \eqref{eq:fpgopt} to hold. 
Sufficiency, meanwhile, is the statement that $\sigma^\star_C$ is the optimizer on the righthand side. 
We can show this by formulating the optimization as a semidefinite program and finding a matching upper bound using the dual program. 

In particular, following~\cite{watrous_semidefinite_2009}, the optimal value of the (primal) semidefinite program
\begin{align}
\begin{array}{r@{\,\,}rl}
\gamma=&\sup & \tr \, W_{ACA'C'}\tau_{ACA'C'}\\
&\text{s.t.} & \tr_{A'C'} W_{ACA'C'}\leq \id_A\otimes \sigma_C\\
&& \tr \,\sigma_C\leq 1\\
&& W_{ACA'C'},\sigma_C \geq 0 \, ,
\end{array}
\end{align}
satisfies $\gamma=\sup_{\sigma\in \cD(C)} F(\tau_{AC},\id_A\otimes \sigma_C)^2$.
Here $A'\simeq A$, $C'\simeq C$, and we take $\tau_{ACA'C'}$ to be the canonical purification of $\tau_{AC}$ as in Section~\ref{sec:pgf}. 
Using Watrous's general form for semidefinite programs we can easily derive the dual, which turns out to be  
\begin{align}
\begin{array}{r@{\,\,}rl}
\beta=&\inf & \mu \\
&\text{s.t.} & Z_{AC}\otimes \id_{A'C'}\geq \tau_{ACA'C'}\\
&&\mu\id_C\geq\tr_A Z_{AC}\\
&& \mu,Z_{AC}\geq 0 \, .
\end{array}
\end{align}


By weak duality $\gamma\leq \beta$, but the following choice of $\mu$ and $Z_{AC}$ gives $\beta=\fpg(\tau_{AC},\id_A\otimes \sigma^\star_C)^2$ and therefore \eqref{eq:fpgopt}:
\begin{align}
\mu^\star = \left(\tr\, \sqrt{\tau_{AC}}\sqrt{\id_A\otimes \sigma^\star_C}\right)^2\quad \text{and} \quad
Z_{AC}^\star=\tr\left(\sqrt{\tau_{AC}}\sqrt{\id_A\otimes \sigma^\star_C}\right){\tau_{AC}^{\nicefrac12}}\left(\id_A\otimes \sigma^\star_C\right)^{-\nicefrac12}.
\end{align}
Here the inverse of $\id_A\otimes\sigma^\star_C$ is taken on its support.  To see that the first feasibility constraint is satisfied, start with the operator inequality 
\begin{align}
\id_{ACA'C'} \tr \, \sqrt{\tau_{AC}}\sqrt{\id_A\otimes \sigma^\star_C}\geq \left(\tau_{AC}^{\nicefrac14}\left(\id_A\otimes \sigma^\star_C\right)^{\nicefrac14} \otimes \id_{A'C'} \right)\Omega_{ACA'C'}\left(\tau_{AC}^{\nicefrac14}(\id_A\otimes \sigma^\star_C)^{\nicefrac14} \otimes \id_{A'C'} \right),
\end{align}
which holds because the righthand side is the canonical purification of the positive operator $\sqrt{\tau_{AC}}\sqrt{\id_A\otimes \sigma^\star_C}$ and the trace factor on the left is its normalization. 
Conjugating both sides by $\tau_{AC}^{\nicefrac14}(\id_A\otimes \sigma^\star_C)^{-\nicefrac14} \otimes \id_{A'C'}$ preserves the positivity ordering and gives 
\begin{align}\label{eq:feasibility_constraint}
\tr\left(\sqrt{\tau_{AC}}\sqrt{\id_A\otimes \sigma^\star_C}\right) \tau_{AC}^{\nicefrac12}(\id_A\otimes \sigma^\star_C)^{-\nicefrac12}\otimes \id_{A'C'} \geq \left( \tau_{AC}^{\nicefrac12} \otimes \id_{A'C'} \right)\Omega_{ACA'C'} \left( \tau_{AC}^{\nicefrac12}\otimes \id_{A'C'} \right) \, ,
\end{align}
where we used that $\ker \left(\id_A\otimes\sigma^\star_C \right) \subseteq \ker (\tau_{AC})$ (just as in the proof of Lemma~\ref{lem:alphasufficiency}), ensuring that $\tau_{AC}(\id_A\otimes \sigma^\star_C)^{-1} (\id_A\otimes \sigma^\star_C)=\tau_{AC}$. Note that inequality~\eqref{eq:feasibility_constraint} shows that $Z^\star_{AC}\otimes \id_{A'C'}\geq \tau_{ACA'C'}$. Meanwhile, the second constraint is satisfied (with equality in the case where  $\sigma^\star_C$ has full rank) because direct calculation shows that $\tr_A {\tau_{AC}^{\nicefrac12}}(\id_A\otimes \sigma^\star_C)^{-\nicefrac12}\leq\id_C \, \tr\, \sqrt{\tau_{AC}}\sqrt{\id_A\otimes \sigma^\star_C}\,$.

\end{appendices}


 \printbibliography[heading=bibintoc,title=References]

\end{document}